\documentclass{article}
\usepackage{graphicx} 
\usepackage{amsmath,amssymb,amsthm}
\usepackage{authblk}
\usepackage{cleveref} 
\usepackage{complexity} 
\usepackage{thm-restate} 
\usepackage{fullpage} 
\usepackage{xcolor} 
\usepackage{xspace} 

\definecolor{myred}{HTML}{db3f3d}
\definecolor{myblue}{HTML}{0065BD}

\usepackage{tikz}
\usetikzlibrary{positioning,arrows.meta,snakes,calc,fit,decorations.markings}

\usepackage{algorithm}
\usepackage{algpseudocodex}

\usepackage[round]{natbib}

\newtheorem{theorem}{Theorem}[section]

\newtheorem{observation}[theorem]{Observation}
\newtheorem{claim}[theorem]{Claim}

\crefname{claim}{claim}{claims}
\Crefname{claim}{Claim}{Claims}

\theoremstyle{definition}

\newtheorem{example}[theorem]{Example}

\newenvironment{claimproof}
{
    
    \proof
}
{
    \endproof
    
}

\newcommand{\lb}{\lambda} 
\newcommand{\ub}{\mu} 
\newcommand{\boundpair}{(\lb,\ub)} 
\newcommand{\vf}{v} 
\newcommand{\uf}{u} 
\newcommand{\costr}{\mathfrak C} 
\newcommand{\SW}[1][\costr]{\mathcal{SW}(#1)} 
\newcommand{\bs}{\setminus} 
\newcommand{\genA}{a} 
\newcommand{\genB}{b} 
\newcommand{\genGrkA}{\sigma} 
\newcommand{\genGrkB}{\tau} 

\newcommand{\topset}[3][\genA]{\text{top}_{#1}(#2,#3)} 

\newcommand{\feastab}{$^*$\xspace} 
\newcommand{\etc}{\textsc{Exact Cover by $3$-Sets}}
\newcommand{\etcs}{\textsc{X3C}}
\newcommand{\etcS}{\mathfrak S}	
\newcommand{\etcpara}{z} 
\newcommand{\mmm}{\textsc{Minimum Maximal Matching}}
\newcommand{\mmms}{\textsc{MMM}}

\makeatletter
\newcommand*{\grk}[1]{%
  \ifcase#1\or \alpha\or\beta\or\gamma\or\delta\or\epsilon\or\zeta\else\@ctrerr\fi
}
\makeatother

\newcommand{\convexpath}[2]{
[   
    create hullnodes/.code={
        \global\edef\namelist{#1}
        \foreach [count=\counter] \nodename in \namelist {
            \global\edef\numberofnodes{\counter}
            \node at (\nodename) [draw=none,name=hullnode\counter] {};
        }
        \node at (hullnode\numberofnodes) [name=hullnode0,draw=none] {};
        \pgfmathtruncatemacro\lastnumber{\numberofnodes+1}
        \node at (hullnode1) [name=hullnode\lastnumber,draw=none] {};
    },
    create hullnodes
]
($(hullnode1)!#2!-90:(hullnode0)$)
\foreach [
    evaluate=\currentnode as \previousnode using \currentnode-1,
    evaluate=\currentnode as \nextnode using \currentnode+1
    ] \currentnode in {1,...,\numberofnodes} {
-- ($(hullnode\currentnode)!#2!-90:(hullnode\previousnode)$)
  let \p1 = ($(hullnode\currentnode)!#2!-90:(hullnode\previousnode) - (hullnode\currentnode)$),
    \n1 = {atan2(\y1,\x1)},
    \p2 = ($(hullnode\currentnode)!#2!90:(hullnode\nextnode) - (hullnode\currentnode)$),
    \n2 = {atan2(\y2,\x2)},
    \n{delta} = {-Mod(\n1-\n2,360)}
  in 
    {arc [start angle=\n1, delta angle=\n{delta}, radius=#2]}
}
-- cycle
}

\title{Single-Deviation Stability in Additively Separable Hedonic Games with Constrained Coalition Sizes}
\date{}

\author[1]{Martin Bullinger}
\author[2]{Adam Dunajski}
\author[3]{Edith Elkind}
\author[4]{Matan Gilboa}

\affil[1]{ \small School of Engineering Mathematics and Technology, University of Bristol, UK}
\affil[2]{ \small School of Mathematics, University of Edinburgh, UK}
\affil[3]{School of Engineering,  Northwestern University, USA}
\affil[4]{ \small Department of Computer Science, University of Oxford, UK\protect\\ \vspace*{0.1cm} martin.bullinger@bristol.ac.uk, a.dunajski@sms.ed.ac.uk, matan.gilboa@cs.ox.ac.uk, edith.elkind@northwestern.edu}

\begin{document}
\maketitle

\begin{abstract}
We study stability in additively separable 
    hedonic games when coalition sizes have to respect fixed size bounds.
    We consider four classic notions of stability based on single-agent deviations, namely, Nash stability, individual stability, contractual Nash stability, and contractual individual stability.
    For each stability notion, we consider two variants:
    in one, the coalition left behind by a deviator must still be of a valid size, and in the other there is no such constraint.
    We provide a full picture of the existence of stable outcomes with respect to given size parameters.
    Additionally, when there are only upper bounds, we fully characterize the computational complexity of the associated existence problem.     In particular, we obtain polynomial-time algorithms for contractual individual stability and contractual Nash stability, where the latter requires an upper bound of~$2$.
    We obtain further results for Nash stability and contractual individual stability, when the lower bound is at least~$2$. 
\end{abstract}

\section{Introduction}

Imagine you are tasked with splitting undergraduate students into groups for a collaborative coursework project. 
Clearly, the groups should not be too large.
Moreover, the students need to gain experience in working as a team, and thus there is also a lower bound on the size of each group. 
Ideally, the partition of the students into groups should also be robust to students wishing to switch groups, as any such switch contributes to administrative overhead.

Similar scenarios arise in several other settings. 
These include assigning desks to faculty in a department with multi-person offices, assigning participants to hiking groups on Duke of Edinburgh hiking expeditions (which are required to have hiking groups of size between four and seven participants),
organizing seating plans for conference dinners and so on.

Such scenarios can be captured by the framework of \textit{hedonic games}, first presented by \citet{DrGr80a}.
We study the combination of the two key aspects of the above examples, namely prominent solution concepts of stability to deviations and the natural constraint of bounds on coalition sizes.
In particular, the latter has so far received little attention in the widely used framework of additively separable utilities \citep{BoJa02a}, where prior work has primarily focused on upper bounds on the coalition size, for which group stability and welfare maximization have been studied \citep{LHSA24a,FGM25a}.

By contrast, we investigate additively separable hedonic games under stability concepts involving deviations by single agents, including Nash stability (NS), individual stability (IS), contractual Nash stability (CNS), and contractual individual stability (CIS).
We only allow deviations into coalitions of size smaller than the upper bound, i.e., the welcoming coalition must be feasible after the deviation, which is a natural assumption. 
However, it is not as clear whether the size of the abandoned coalition should be constrained in the same way: On the one hand, a selfish agent might not care about the feasibility of the abandoned coalition, while on the other hand, in some settings such deviations may be prohibited (e.g., students are only allowed to leave their project group if the remaining group is large enough). 
To capture this distinction, we introduce two variants of each stability notion: e.g., NS\feastab denotes a Nash-stable partition where a Nash-deviation must maintain the feasibility of the abandoned coalition, while NS denotes one where deviations may render the abandoned coalition infeasible. 

We begin with some preliminary considerations about the existence of coalition structures that satisfy the size constraints.
We provide a simple characterization and apply it to show that, for fixed size constraints, such outcomes always exist for a sufficiently large number of agents.
Our main objective is to then understand the existence and computational complexity of stable outcomes.
We first present a complete classification of existence for all mentioned stability concepts. 
Akin to similar settings, symmetric valuations lead to existence of NS\feastab, the strongest feasible stability concept.
By contrast, even for symmetric $0/1$-valuations, CIS outcomes, i.e., the weakest standard stability notion, may not exist.
Moreover, for nonsymmetric valuations, only CIS\feastab outcomes are guaranteed to exist.

We continue with establishing a complete complexity picture for the setting where coalition sizes are only constrained by an upper bound $\ub$.
In particular, we provide a polynomial-time algorithm for CIS,
amending the algorithm in the existing literature by \citet{ABS11c}. 
In addition, we present a polynomial-time algorithm for the construction of CNS outcomes when $\ub = 2$.
For all other upper bounds (and a lower bound $\lb$ of~$1$) and all other stability concepts, we obtain \NP-completeness of the existence of stability.

We conclude with results for a nontrivial lower bound $\lb \ge 2$.
We show \NP{}-completeness for Nash stability for any size constraints satisfying $\ub \ge 4$ and $\lb < \ub$.
Finally, we present polynomial-time algorithms for CIS\feastab for any $\lb$ and $\ub$ when the additively separable valuations are nonzero or nonnegative.

\section{Related Work}

Hedonic games were first introduced by \citet{DrGr80a}, and started to receive continuous and extensive attention since the introduction of additively separable hedonic games (ASHGs) by \citet{BoJa02a} two decades later.
The book chapters by \citet{AzSa15a} and \citet{BER24a} present introductory texts.

A main goal in hedonic games is to consider the existence and computability of stable outcomes.
The stability notions based on single-agent deviations that we consider are well understood for ASHGs under unconstrained coalition sizes.
The existence problem for NS, IS, and CNS is known to be \NP-complete \cite{SuDi10a,BBT23a}, while CIS admits a polynomial-time algorithm \cite{ABS11c}.\footnote{Unfortunately, we identify an inaccuracy in this algorithm that we discuss in \Cref{app:CISflawed} and rectify in our treatment of CIS.}
Notions of group stability as measured by the core or strict core, and popularity are even known to be $\Sigma_2^p$-complete \cite{Woeg13a,BuGi25a}.
By contrast, some positive results are known, in particular for restricted domains of valuations.
First, symmetric valuations lead to Nash-stable outcomes \cite{BoJa02a}, a common result in hedonic games that we will extend to coalition size bounds.
Moreover, if the range of valuations only has one nonpositive or one nonnegative value, then existence of IS and CNS are guaranteed \cite{BBT23a}.
These restrictions encompass, in particular, subclasses of ASHGs based on the distinction of friends and enemies, for which even group-stable outcomes exist \cite{DBHS06a}.
Interestingly, our counterexamples to existence even hold under these preferences restrictions.
Finally, IS and CNS exist and are efficiently computable in random hedonic games, while NS does not exist with high probability. \cite{BuKr24a}.

We next discuss related work on hedonic games with constraints on coalition sizes.
\citet{LHSA24a} study ASHGs under both welfare maximization and stability based on group deviations.
They find hardness for maximizing welfare as well as for the existence of an outcome in the core or strict core.
By contrast, an outcome in the contractual strict core exists and can be computed in polynomial time.
In addition, \citet{FGM25a} provide a more nuanced picture of the complexity of welfare maximality through the lens of parameterized complexity.
Our work complements these works by studying prominent stability notions based on deviations by single agents, while we are also the first to additionally consider a lower bound on coalition sizes.

In principle, ASHGs with coalition sizes can be modeled in fully expressive preference models, e.g., when explicitly listing coalition rankings \citep{Ball04a} or encoding them with Boolean formulae \citep{ElWo09a}.
However, since the reduced games in their hardness constructions do not fulfill our size constraints, their proofs do not extend to our setting.
Another way to capture arbitrary hedonic games is via hedonic diversity games with an arbitrary number of colors \citep{GHK+23a}. 
However, their positive results mostly hold for a small number of colors (with the exception of a lemma concerning parameterized complexity) and thus do not have implications for our work.

Moreover, ASHGs with upper-bounded coalition sizes are a special case of topological distance games \cite{BuSu24a,DEKS24a}.
There, players are assigned to the vertices of a topology graph and obtain preferences based on additively separable valuations that are discounted by the distance on the topology graph.
The notion of jump stability for topological distance games is akin to NS for ASHGs.
For this notion, \citet{BuSu24a} obtain positive results for special valuations and topology graphs, most of which do not correspond to ASHGs. The only result that applies to our setting is for the severe restriction of nonnegative valuations where positive valuations form an acyclic graph.
Furthermore, upper bounds on coalition sizes were also considered in an online model of ASHGs \cite{FMM+21a,CoAg25a}.

We conclude with more loosely related models. 
Somewhat similar to bounded coalition sizes, \citet{BMM22a} require outputs of fixed coalition sizes.
However, this means that coalitions are always full and, instead of single agents deviating unilaterally, stability is based on swaps.
Another related question that implicitly enforces coalition bounds is to ask for a partition into a given number of coalitions \citep{LMNS23a,DEI+25a,AARS25a}.
This leads to lower bounds in particular when partitions have to be balanced, i.e., coalitions have to be of similar size.
However, in this model the coalition size bounds depend on the total number of agents, while our bounds are global parameters.
Note that \citet{LMNS23a} and \citet{AARS25a} consider a restricted model of simple, i.e., unweighted, ASHGs.
Moreover, these works consider different solution concepts, mostly inspired by envy-freeness in the fair division literature.
Finally, beyond ASHGs, \citet{DDDS22a} consider coalition bounds for group activity selection, a model related to anonymous hedonic games \citep{BoJa02a}.

\section{Preliminaries}

In this section, we define our model.
Given a positive integer $i\in \mathbb N$, we use the notion $[i]:=\{1,\dots, i\}$.
Moreover, for two integers $i,j\in \mathbb Z$ and $k\in \mathbb N$, we use the notation $i \equiv_k j$ to denote equivalence modulo $k$.

\subsection{Hedonic Games}
Coalition formation is concerned with partitioning a set of agents into disjoint coalitions according to their preferences.
We consider a finite set of agents $N$.
A nonempty subset of $N$ is called a \emph{coalition}.
We want coalitions to be within given \emph{size bounds}.
Let $\lb, \ub\in \mathbb N$ with $\lb \le \ub$.
A coalition $C\subseteq N$ with $\lb\le |C|\le \ub$ is called a $(\lb,\ub)$-coalition. 
If $|C|=\ub$, we say that $C$ is \emph{full}.
A \emph{coalition structure} or \emph{partition} of $N$ is a set of pairwise disjoint coalitions whose union is $N$.
Given a partition~$\costr$, we denote by $\costr(\genA)$ the coalition containing agent~$\genA$.
A partition is called a $(\lb,\ub)$-partition if it consists of $(\lb,\ub)$-coalitions only.

Let $\mathcal N_{\genA}$ denote all possible coalitions containing agent $\genA$, i.e., $\mathcal N_{\genA} :=\{C\subseteq N: \genA\in C\}$.
A \emph{hedonic game} is defined by a tuple $(N,\succsim)$, where $N$ is an agent set and ${\succsim} = (\succsim_{\genA})_{\genA\in N}$ is a tuple of weak orders $\succsim_{\genA}$ over $\mathcal N_{\genA}$.
The weak order $\succsim_{\genA}$ represents the preferences of agent $\genA$. 
Let $C,C'\in \mathcal N_{\genA}$.
We say that $\genA$ \emph{weakly prefers} $C$ over $C'$ if $C\succsim_{\genA} C'$.
The strict part of $\succsim_{\genA}$ is denoted by $\succ_{\genA}$, i.e., $C\succ_{\genA} C'$ if and only if $C \succsim_{\genA} C'$ and not $C' \succsim_{\genA} C$.
We say that $\genA$ \emph{prefers} $C$ over $C'$ if $C\succ_{\genA} C'$.
We extend the preference order to partitions by setting $\costr\succsim_{\genA}\costr'$ if and only if $\costr(\genA) \succsim_{\genA} \costr'(\genA)$. 

An important subclass of hedonic games are additively separable hedonic games as first considered by \citet{BoJa02a}.
They are encoded by a pair $(N,\vf)$, where $N$ is an agent set and $\vf = (\vf_{\genA})_{\genA\in N}$ is a tuple of utility functions $\vf_{\genA}\colon N\setminus\{\genA\}\rightarrow\mathbb{R}$.
The \emph{(additively separable) utility} of an agent $\genA\in N$ for coalition $C\in \mathcal N_{\genA}$ is defined as $\uf_{\genA}(C) :=\sum_{\genB\in C\setminus \{\genA\}} \vf_{\genA}(\genB)$.
Note that this implies $\uf_{\genA}(\{\genA\}) = 0$.
Utilities again extend to partitions by setting $\uf_{\genA}(\costr) :=\uf_{\genA}(\costr(\genA))$.
The \emph{additively separable hedonic game} (ASHG) induced by $(N,\vf)$ is the hedonic game $(N,\succsim)$, where for all $\genA\in N$ and $C,C'\in \mathcal N_{\genA}$, it holds that
\begin{equation*}
    C\succsim_{\genA} C'\text{ if and only if }\sum_{\genB\in C\setminus \{\genA\}}\vf_{\genA}(\genB) \geq \sum_{\genB\in C'\setminus \{\genA\}}\vf_{\genA}(\genB)\text.
\end{equation*}
Hence, $C\succsim_{\genA} C'$ if and only if $\uf_{\genA}(C) \ge \uf_{\genA}(C')$.

Every ASHG can be naturally represented by a complete directed graph $G=(N,E)$ with weight $\vf_{\genA}(\genB)$ on arc $(\genA,\genB)$. An ASHG is said to be \emph{symmetric} if $\vf_{\genA}(\genB)=\vf_{\genB}(\genA)$ for all $\genA, \genB\in N$, 
and can be represented by a complete undirected graph with weight $\vf_{\genA}(\genB)$ on edge $\{\genA,\genB\}$.
An ASHG is said to be \emph{simple} if $\vf_{\genA}(\genB) \in \{0,1\}$ for all $\genA, \genB\in N$.
Simple symmetric ASHGs can be represented by an unweighted undirected graph. 
Finally, an important quantity is the sum of all agents utilities as captured by the so-called (utilitarian) social welfare.
Formally, given a partition $\costr$, its \emph{social welfare} $\SW$ is defined by $\SW := \sum_{\genA\in N}\uf_{\genA}(\costr)$.

\subsection{Useful Agent Sets}\label{sec:agentsets}

We continue with some notation to refer to certain sets of agents that will be important in our algorithms.
Assume that we are given an ASHG $(N,\vf)$ together with an agent $\genA\in N$.

We first introduce a notation for choosing a subset of most preferred agents among a given subset of agents.
Let $A\subseteq N$ and $k\in \mathbb N$.
If $k > |A\setminus \{\genA\}|$, we set $\topset{k}{A} = A\setminus \{\genA\}$.
Otherwise, let $\topset{k}{A}$ be a subset of $k$ agents among $A\setminus \{\genA\}$ that maximizes $\genA$'s utility among subsets of size $k$ from $A\setminus \{\genA\}$. We explicitly allow that $\genA\in A$ but $\genA$ is never contained in $\topset{k}{A}$.
Moreover, we can easily find $\topset{k}{A}$ in time polynomial in $n = |N|$: simply order agents in $A\setminus \{\genA\}$ by decreasing valuation for $\genA$ and take the first $k$ agents in the obtained order.
Note that the obtained set is not necessarily unique, but our algorithms work for any such set.

Next, we want a notation for agents valued positively and negatively by a single agent or set of agents.
Given $\genA,\genB\in N$, we say that $\genB$ is a \emph{friend} (or \emph{enemy}) of $\genA$ if $\vf_{\genA}(\genB) > 0$ (or $\vf_{\genA}(\genB) < 0$).
Moreover, given an agent $\genA\in N$ and subsets $A,B\subseteq N$, we define
\begin{itemize}
        \item $\text{Fr}(\genA,B) := \{ \genB\in B \colon \vf_{\genA}(\genB) > 0\}$,
        \item $\text{Fr}(A,B) := \bigcup_{\genA\in A} \text{Fr}(\genA,B) = \{ \genB\in B \colon \exists \genA\in A \text{ such that } \vf_{\genA}(\genB) > 0\}$,
        \item $\text{En}(\genA,B) := \{ \genB\in B \colon \vf_{\genA}(\genB) < 0\}$, and
        \item $\text{En}(A,B) := \bigcup_{\genA\in A} \text{En}(\genA,B) = \{ \genB\in B \colon \exists \genA\in A \text{ such that } \vf_{\genA}(\genB) < 0\}$.
    \end{itemize}

Hence, these sets capture the friends and enemies of $\genA$ and $A$ in $B$.
Note that agents $\genB$ with $\vf_{\genA}(\genB)=0$ are neither friends nor enemies.

\subsection{Single-Deviation Stability}

Stability captures the absence of beneficial deviations by agents joining other coalitions or forming a new coalition.
We focus on stability notions that are concerned with the incentives of single agents to deviate.
A \emph{single-agent deviation} performed by agent $\genA$
transforms a partition $\costr$ into a partition $\costr'$ 
where $\costr(\genA)\neq\costr'(\genA)$ and, for all agents $\genB\neq \genA$, it holds that $\costr(\genB)\setminus\{\genA\} = \costr'(\genB)\setminus\{\genA\}$.
We write $\costr \xrightarrow{\genA} \costr'$ to denote such a single-agent deviation.

In the presence of size bounds, a deviating agent should, however, only end up in a coalition that obeys the bounds.
A single-agent deviation $\costr \xrightarrow{\genA} \costr'$ is said to be \emph{$(\lb,\ub)$-permissible} if $\lb\le |\costr'(\genA)| \le \ub$.
Hence, an agent is only allowed to form a new coalition if $\lb = 1$, and an agent can only join a nonempty coalition of size at most $\ub - 1$.

We now introduce standard stability concepts based on single-agent deviations \cite{BoJa02a,SuDi07b}.
However, in contrast to the known definitions for partitions without size constraints, our definition of stability is based on permissible deviations.
We assume that we are given a fixed hedonic game $(N,\succsim)$ and size bounds $\lb, \ub\in \mathbb N$ with $\lb \le \ub$.

A \emph{Nash deviation} is a single-agent deviation $\costr \xrightarrow{\genA} \costr'$ such that $\costr'(\genA)\succ_{\genA} \costr(\genA)$.
A $(\lb,\ub)$-partition $\costr$ is said to be \emph{Nash-stable} (NS) if no agent can perform a $(\lb,\ub)$-permissible Nash deviation.

Since Nash stability disregards the preferences of agents apart from the deviating agent, various refinements have been proposed which additionally require the consent of the abandoned and the joined agents. 
An \emph{individual deviation} is a Nash deviation $\costr \xrightarrow{\genA} \costr'$ such that $\{\genB\in \costr'(\genA)\colon \costr(\genB) \succ_{\genB} \costr'(\genB)\} = \emptyset$.
Similarly, a \emph{contractual deviation} is a Nash deviation $\costr \xrightarrow{\genA} \costr'$ such that $\{\genB\in \costr(\genA)\colon \costr(\genB) \succ_{\genB} \costr'(\genB)\} = \emptyset$. 
Moreover, $\costr \xrightarrow{\genA} \costr'$ is called a \emph{contractual individual deviation} if it is both an individual and a contractual deviation.
A $(\lb,\ub)$-partition is said to be \emph{individually stable} (IS), \emph{contractually Nash-stable} (CNS), or \emph{contractually individually stable} (CIS) if it allows for no $(\lb,\ub)$-permissible individual, contractual, or contractual individual deviation, respectively.
Since stability against a superset of deviations is harder to satisfy, NS implies IS and CNS, both of which in turn imply CIS.
We refer to NS, IS, CNS, and CIS as \emph{standard stability concepts}.

Still, even though a $(\lb,\ub)$-permissible single-agent deviation results in a $(\lb,\ub)$-coalition for the deviator, it can result in a partition that is not a $(\lb,\ub)$-partition, because the abandoned coalition might have been of size exactly $\lb$ and $\lb\ge 2$.
We say that a single-agent deviation $\costr \xrightarrow{\genA} \costr'$ is \emph{$(\lb,\ub)$-feasible} if $\costr'$ is a $(\lb,\ub)$-partition.
Hence, $(\lb,\ub)$-feasible deviations require that a coalition can only be abandoned if it remains within the size bounds, i.e., if it has a size of at least $\lb + 1$ before it is abandoned (unless $\lb = 1$).

Each of our stability concepts gives rise to a related concept based on feasible instead of permissible deviations.
We speak of the associated \emph{feasible} stability concepts and indicate them by a star in the abbreviated notation, e.g., we say feasible Nash stability and denote it by NS\feastab.
As feasible stability concepts prevent a subset of deviations, they result in the logically weaker stability concepts, e.g., NS implies NS\feastab. 
An overview of our stability concepts is given in \Cref{fig:concepts}.
Arrows indicate logical relationships between them.
It is easy to see that apart from the depicted implications and the ones implied by transitivity, there are no further relationships.
For example, there are NS\feastab partitions that are no CIS partition (cf.~\Cref{ex:intro}).
For all abbreviations $\alpha$ of our solution concepts, e.g., $\alpha = \text{NS}$, we speak of an $\alpha$ deviation to refer to a deviation associated with the solution concept, e.g., NS deviation.

\begin{figure}[tb]
    \centering
    \begin{tikzpicture}
    \pgfmathsetmacro\yscale{1.2}
    \pgfmathsetmacro\xmove{2.8}
    \pgfmathsetmacro\ymove{.75}
        \node (FCIS) at (0,0) {CIS\feastab};
        \node (FIS) at (-1.1,\yscale) {IS\feastab};
        \node (FCNS) at (1.1,\yscale) {CNS\feastab};
        \node (FNS) at (0,2*\yscale) {NS\feastab};
        \node (CIS) at ($(0,0)+(-\xmove,\ymove)$) {CIS};
        \node (IS) at ($(-1.1,\yscale)+(-\xmove,\ymove)$) {IS};
        \node (CNS) at ($(1.1,\yscale)+(-\xmove,\ymove)$) {CNS};
        \node (NS) at ($(0,2*\yscale)+(-\xmove,\ymove)$) {NS};

        \draw[->]   (FNS)  edge (FCNS)
                    (FNS)  edge (FIS)
                    (FCNS) edge (FCIS)
                    (FIS)  edge (FCIS)
                    (NS)  edge (CNS)
                    (NS)  edge (IS)
                    (CNS) edge (CIS)
                    (IS)  edge (CIS)
                ;
        \foreach \S in {CIS,IS,CNS,NS}
        {\draw[->] (\S) edge (F\S);}
    \end{tikzpicture}
    \caption{Overview and logical implications of stability concepts.
    Each standard stability concept is associated with a feasible stability concept indicated by a star.}
    \label{fig:concepts}
\end{figure}
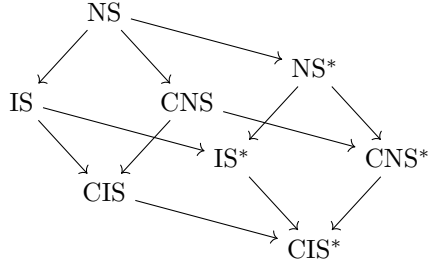

Finally, a fundamental desideratum beyond stability is that an agent wants to end up in a coalition that is at least as good as being on their own.
A coalition $C\in \mathcal N_{\genA}$ is said to be \emph{individually rational} for agent $\genA$ if $C\succsim_{\genA} \{\genA\}$.
In the case of an ASHG, this requires that $\uf_{\genA}(C)\ge 0$.
Finally, a partition $\costr$ is individually rational if for 
all $\genA\in N$, $\costr(\genA)$ is individually rational.
Since IS and IS\feastab deviations do not prevent an agent from abandoning a coalition with negative utility, IS and IS\feastab imply individual rationality if $\lb = 1$.
However, when $\lb \ge 2$, we do not even have that NS implies individual rationality. 

\section{Feasible Partitions}\label{sec:feasible}

In contrast to games with unconstrained or upper-bounded coalition sizes, lower bounds may prevent the formation of partitions in which all coalitions satisfy the size constraints. 
For example, a group of $8$ agents cannot be partitioned into coalitions of sizes between $5$ and $7$.
Hence, as a preliminary consideration, we first consider the existence of $(\lb,\ub)$-partitions.

\begin{restatable}{proposition}{feasible}\label{prop:feasible}
    Let $\lb,\ub\in \mathbb N$ with $\lb \le \ub$ and $n\in \mathbb N$. 
    A $(\lb,\ub)$-partition of $n$ agents exists if and only if $n\le \left\lfloor \frac n{\lb}\right\rfloor \ub$.
\end{restatable}

\begin{proof}
	Let $\lb,\ub\in \mathbb N$ with $\lb \le \ub$ and $n\in \mathbb N$. 
	Assume that there exists a $(\lb,\ub)$-partition $\costr$ of $n$ agents.
	As each coalition in $\costr$ is of size at least $\lb$, we know that $|\costr|\le \left\lfloor \frac n{\lb}\right\rfloor$.
	Hence,
	$$n\le \sum_{C\in \costr}|C| \le \sum_{C\in \costr}\ub = |\costr|\ub \le \left\lfloor \frac n{\lb}\right\rfloor \ub\text.$$
	
	Conversely assume that $n$ satisfies $n\le \left\lfloor \frac n{\lb}\right\rfloor \ub$.
	It holds that $n = \lb \frac n {\lb} \ge \lb \left\lfloor \frac n{\lb}\right\rfloor$.
	Hence, we can use $\lb \left\lfloor \frac n{\lb}\right\rfloor$ agents to form $\left\lfloor \frac n{\lb}\right\rfloor$ coalitions of size $\lb$.
	These coalitions have space for $\left\lfloor \frac n{\lb}\right\rfloor \ub$ agents to constitute a $(\lb,\ub)$-partition.
	Hence, since $n\le \left\lfloor \frac n{\lb}\right\rfloor \ub$, they have space for the remaining agents to form a $(\lb,\ub)$-partition.
\end{proof}

Note that the condition $n\le \left\lfloor \frac n{\lb}\right\rfloor \ub$ excludes the case where $n< \lb$ (in which case we could not form any coalition of size at least $\lb$), as then this condition reads $n\le 0$, which is false for any $n\in \mathbb N$.
Also, $\left\lfloor \frac n{\lb}\right\rfloor$ is the largest number of coalitions that can be created with $n$ agents so that each coalition has size at least~$\lb$.
Hence, \Cref{prop:feasible} implies that there exists a $(\lb,\ub)$-partition of $n$ agents if and only if there exists such a partition into the largest possible number of coalitions.
In fact, we can use it to characterize precisely for which numbers $k$ there exists a $\boundpair$-partition into $k$ coalitions.

\begin{restatable}{proposition}{feasibleKpart}\label{prop:feasibleKpart}
    Let $\lb,\ub\in \mathbb N$ with $\lb \le \ub$ and $k,n\in \mathbb N$. 
    Then there exists a $(\lb,\ub)$-partition of $n$ agents into $k$ coalitions if and only if $k \lb\le n\le k \ub$. \end{restatable}

\begin{proof}
	Assume first that there exists a $(\lb,\ub)$-partition $\costr$ of $n$ agents into $k$ coalitions.
	As each coalition in $\costr$ is of size at least $\lb$, we know that 
		$k\lb = k|\costr| \le n$.
	Moreover, 
	\begin{equation*}
		n \le \sum_{C\in\costr}|C| \le \sum_{C\in\costr}\ub = |\costr|\ub = k\cdot\ub\text.
	\end{equation*}
	
	Conversely, assume that $k\lb \le n\le k \ub$.
    Then, since $k$ is an integer, $k \le \frac n{\lb}$ implies $k \le \left\lfloor \frac n{\lb}\right\rfloor$.
	Combining this with $n\le k\ub$ implies that $n\le \left\lfloor \frac n{\lb}\right\rfloor \ub$.
	Hence, by \Cref{prop:feasible}, there exists some $(\lb,\ub)$-partition $\costr$ of $n$ agents, say $\costr = \{C_1,\dots, C_{\ell}\}$, where $\ell = |\costr|$.
	Assume that $|\costr|>k$.
	Then, since $n \le k \ub$, we have $k\ub - \sum_{i = 1}^k|C_i| \ge n - \sum_{i = 1}^k|C_i| = \sum_{i = k+1}^{\ell}|C_i|$.
	Hence, there is enough space in the first $k$ coalitions to add the agents in the coalitions $C_{k+1},\dots, C_{\ell}$ to obtain a $(\lb,\ub)$-partition of the $n$ agents into $k$ coalitions.
	\end{proof}

Additionally, \Cref{prop:feasible} can be applied to prove an interesting consequence about the existence of $(\lb,\ub)$-partitions.
Clearly, if $\lb = \ub$, then a $(\lb,\ub)$-partition exists if and only if $\ub$ divides $n$.
Otherwise, we now show that a $(\lb,\ub)$-partition exists for a sufficiently large number of agents.

\begin{restatable}{proposition}{LargeFeasible}\label{prop:LargeFeasible}
    Let $\lb,\ub\in \mathbb N$ with $\lb < \ub$. 
    Then a $(\lb,\ub)$-partition exists whenever $n \ge \frac {\lb - 1}{\ub - \lb} \ub$.
\end{restatable}

\begin{proof}
	Let $\lb < \ub$. 
	Assume that $n \ge \frac {\lb - 1}{\ub - \lb} \ub$.
	Multiplication by $\frac {\ub-\lb}{\lb}$ implies that 
	\begin{equation}\label{eq:feas:bound}
		\frac {n}{\lb} \ub - n = \frac{\ub-\lb}{\lb} n \ge \frac {\lb - 1}{\lb} \ub \Leftrightarrow \left[\frac n{\lb} - \frac{\lb - 1}{\lb}\right]\cdot \ub \ge n\text.
	\end{equation}
	
	Hence, 
	$$\left\lfloor \frac n{\lb}\right\rfloor \cdot \ub 
	= \left[\frac n{\lb} + \left(\left\lfloor \frac n{\lb}\right\rfloor - \frac n{\lb}\right)\right]\cdot \ub 
	\ge \left[\frac n{\lb} - \frac{\lb - 1}{\lb}\right]\cdot \ub \ge n\text.$$
	For the first inequality, note that the remainder of division by $\lb$ is an integer in $\{0,1,\dots, \lb -1\}$.
	Hence, rounding down $\frac n{\lb}$ can decrease this number by at most $\frac{\lb - 1}{\lb}$.
	In the second inequality, we applied \Cref{eq:feas:bound}.
	Thus, by \Cref{prop:feasible}, there exists a feasible partition.
	\end{proof}

While \Cref{prop:feasible} allows us to check for the existence of a $(\lb,\ub)$-partition with a few elementary operations, \Cref{prop:LargeFeasible} implies that the nonexistence of $(\lb,\ub)$-partitions is not an algorithmic issue if we fix size bounds.
For instance, since hardness reductions construct games with an unbounded number of agents, our hardness results for stability do not rely on the nonexistence of stability for small agent numbers (in fact, all games in our reductions admit partitions for the considered size bounds).

\section{Existence of Stable Partitions}\label{sec:existence}

We turn to the consideration of stability.
We start with an example that highlights basic relationships of stability concepts once we have a nontrivial lower bound.

\begin{example}\label{ex:intro}
    Let $\lb = 2$, $\ub \ge 3$, and $k\in \mathbb N$. 
    We define two ASHGs $(N,\vf)$ and $(N,\vf')$ where $N = \{\genA_{i},\genB_{i}\colon i\in [k]\}$.
    Symmetric valuations $\vf$ are given by $\vf(\genA_{i},\genB_{i}) = -1$ for all $i\in [k]$ and $\vf(\genA,\genB) = 1$ for any other pair of agents.
    Symmetric valuations $\vf'$ are given by $\vf'(\genA,\genB) = -1$ for all pairs of agents $\genA,\genB\in N$.
    We illustrate $(N,\vf)$ in \Cref{fig:intro}.
    For both games, we consider the partition $\costr = \{\{\genA_{i},\genB_{i}\}\colon i\in [k]\}$.

    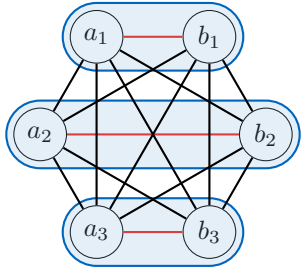
\begin{figure}[tb]
        \centering
        \begin{tikzpicture}[every node/.style={draw, circle, minimum size=.7cm, inner sep=0pt}]

            \pgfmathsetmacro\picscale{1.5}
            \node (a1) at (120:\picscale) {$\genA_1$};
            \node (a2) at (180:\picscale) {$\genA_2$};
            \node (a3) at (240:\picscale) {$\genA_3$};
            \node (b1) at (60:\picscale) {$\genB_1$};
            \node (b2) at (0:\picscale) {$\genB_2$};
            \node (b3) at (300:\picscale) {$\genB_3$};

            \foreach \i in {1,2,3}
            {
            \draw[thick,myblue, fill=myblue!50, fill opacity=0.2] \convexpath{a\i,b\i}{0.45cm};
            \draw[myred,thick] (a\i) edge (b\i);
            }
            \foreach \i/\j in {1/2,2/3,3/1}
            {
            \draw[thick] (a\i) edge (a\j);
            \draw[thick] (b\i) edge (b\j);
            }
            \foreach \i/\j in {1/2,1/3,2/1,2/3,3/1,3/2}
            {
            \draw[thick] (a\i) edge (b\j);
            }
        \end{tikzpicture}
        \caption{Illustration of the game $(N,\vf)$ in \Cref{ex:intro} for the case $k = 3$.
        Black and red edges indicate a symmetric valuation of $1$ and $-1$, respectively.
        The partition $\costr$ is indicated in blue.}
        \label{fig:intro}
    \end{figure}

    For the ASHG $(N,\vf)$, $\costr$ is an NS\feastab $(\lb,\ub)$-partition because no agent can perform a $(\lb,\ub)$-feasible deviation when $\lb = 2$.
    However, $\costr$ is not a CIS $(\lb,\ub)$-partition as each deviation of any agent to any other nonempty coalition is a $(\lb,\ub)$-permissible CIS deviation.
    Hence, even our strongest feasible stability concept does not imply our weakest standard stability concept.

    For our second ASHG $(N,\vf')$, $\costr$ clearly is an NS $(\lb,\ub)$-partition.
    Indeed, every $(\lb,\ub)$-permissible deviation is to join another nonempty coalition which is always worse for the deviator.
    Hence, even our strongest stability concept does not imply individually rationality for any agent.
\end{example}

For the rest of this section, we are concerned with the existence of stable partitions.
Our first observation is that feasible stability notions can be satisfied under symmetric valuations.
This is a standard result for unconstrained coalition sizes \cite{BoJa02a} and has also been observed for upper-bounded coalition sizes \cite{LAH23a,FGM25a} as well as for hedonic games with fixed-sized coalitions and topological distance games \cite{BMM22a,BuSu24a}.
The idea is that a partition maximizing social welfare is stable.
For completeness, we include the standard proof.

\begin{restatable}{proposition}{symexist}\label{prop:symexist}
    Let $\lb,\ub\in \mathbb N$ with $\lb\le\ub$.
    Any symmetric ASHG that admits a $(\lb,\ub)$-partition also admits an NS\feastab $(\lb,\ub)$-partition.
\end{restatable}

\begin{proof}
	Consider a symmetric ASHG $(N,\vf)$.
	We prove the 
	Consider a $\boundpair$-partition $\costr^*$ that maximizes social welfare among  $\boundpair$-partitions.
	We claim that $\costr^*$ is an NS\feastab $\boundpair$-partition.
	
	Assume for contradiction that there exists an agent $\genA\in N$ that can perform a $\boundpair$-feasible NS deviation $\costr \xrightarrow{\genA} \costr'$ resulting in a partition $\costr'$.
	Note that the only agents apart from $\genA$ whose utilities change are agents in $\costr(\genA)\setminus \{\genA\}$ (who diminish their utility by their valuation for $\genA$) and in $\costr'(\genA)\setminus \{\genA\}$ (who increase their utility by their valuation for $\genA$)
	We obtain
	\begin{align*}
		&\SW[\costr'] - \SW = \sum_{\genB\in N}(\uf_{\genB}(\costr') - \uf_{\genB}(\costr))\\
		&= \uf_{\genA}(\costr') - \uf_{\genA}(\costr) + \sum_{\genB\in \costr'(\genA)\setminus \{\genA\}} \vf_{\genB}(\genA) + \sum_{\genB\in \costr(\genA)\setminus \{\genA\}} - \vf_{\genB}(\genA)\\
		&= \uf_{\genA}(\costr') - \uf_{\genA}(\costr) + \sum_{\genB\in \costr'(\genA)\setminus \{\genA\}} \vf_{\genA}(\genB) + \sum_{\genB\in \costr(\genA)\setminus \{\genA\}} - \vf_{\genA}(\genB)\\
		&= 2(\uf_{\genA}(\costr') - \uf_{\genA}(\costr)) > 0\text.
	\end{align*}
	
	There, we use symmetry in the second-to-last line.
	Hence $\SW[\costr'] > \SW$, contradicting our maximality assumption.
\end{proof}

By contrast, if we consider $(\lb,\ub)$-permissible instead of $(\lb,\ub)$-feasible deviations, then stability is not even guaranteed for our weakest standard stability notion, even if valuations are simple and symmetric.
Indeed, simple symmetric ASHGs whose valuations evolve from a star graph have this property.

\begin{restatable}{proposition}{nonexistCISsymm}\label{prop:nonexistCISsymm}
        Let $2\le \lb < \ub$. 
    Then there exists a simple symmetric ASHG
    that contains no CIS $(\lb,\ub)$-partition.
\end{restatable}

\begin{proof}
	Let $2\le \lb < \ub$. 
		We will construct an ASHG $(N,\vf)$ with the desired property where $N = \{\genA_{i}\colon i\in [2\lb -1]\}\cup\{c\}$, i.e., the ASHG contains $n = 2\lb$ agents.
	Simple and symmetric valuations are given by $\vf(c,\genA_{i}) = 1$ for all $i\in [2\lb -1]$, and all other valuations are set to~$0$.
	
	Consider any $\boundpair$-partition $\costr$.
	By design of the number of agents, $\costr$ has to consist of two coalitions of size $\lb$. 
	We will call these $S_1$ and $S_2$, i.e., $\costr = \{S_1,S_2\}$.
	
	Without loss of generality $c\in S_1$.
	But then any agent in $S_2$ can perform a $\boundpair$-permissible CIS deviation.
	Hence, $(N,\vf)$ does not admit CIS $\boundpair$-partitions.
\end{proof}

\Cref{prop:symexist,prop:nonexistCISsymm} provide a complete picture about the existence of our stability concepts for symmetric ASHGs.
While even the strongest feasible stability concept can be satisfied, the weakest standard stability concept cannot.
This leads us to the question whether feasible stability is still guaranteed for nonsymmetric valuations.
First, since $(\lb,\ub)$-feasible deviations lead to $(\lb,\ub)$-partitions and increase the social welfare, CIS\feastab partitions are guaranteed to exist.
Hence, any partition maximizing social welfare among $(\lb,\ub)$-partitions is a CIS\feastab partition.
Consequently, the existence of such a partition only depends on the existence of some $(\lb,\ub)$-partition.

\begin{observation}\label{prop:nonexistCIS}
    Let $\lb,\ub\in \mathbb N$ with $\lb\le\ub$.
    Any ASHG that admits a $(\lb,\ub)$-partition also admits a CIS\feastab $(\lb,\ub)$-partition.
\end{observation}

However, existence for stronger feasible stability concepts is not guaranteed anymore.
For IS\feastab, examples are given by instances where the positive valuations form a directed cycle.

\begin{restatable}{proposition}{nonexistISasy}\label{prop:nonexistISasy}
    Let $2\le \lb < \ub$. 
    Then there exists an infinite family of simple ASHGs
    that contain no IS\feastab $(\lb,\ub)$-partition.
            \end{restatable}

\begin{proof}
	Consider any pair of bounds $\lb,\ub\in \mathbb N$ with $2 \le\lb <\ub$ and choose $n\in \mathbb N$ such that $\lb$ and $\ub$ do not divide $n$.\footnote{Without loss of generality, we may choose $n$ large enough to guarantee the existence of $\boundpair$-partitions, see \Cref{prop:LargeFeasible}.}
	We will construct an ASHG $(N, \vf)$ for any such set of parameters.
	Let $N = \{\genA_{i}\colon i\in [n]\}$ and valuations given as
	\begin{align}
		\vf_{\genA_{i}}(\genA_{j}) = \begin{cases}
			1 & \text{if } j \equiv_n i+1,\\
			0 & \text{otherwise.}
		\end{cases}
	\end{align}
	
	Hence, the game corresponds to a directed cycle among $n$ agents.
	Note that as all player values are nonnegative we have IS $=$ NS. 
	
	For any strict subset of agents $N' \subsetneq N$, there must exist an agent $\genA \in N'$ such that there exists an agent $\genB \notin N'$ with $v_{\genA}(\genB) = 1$.
	Hence, the unique agent in the cycle succeeding $\genA$ is not contained.
	Note that $\genA$ would like to perform a IS\feastab deviation to join $\genB$, but this deviation might not be $\boundpair$-feasible.
	We will now show that in any possible $\boundpair$-partition there exists such a situation where a $\boundpair$-feasible deviation exists.
		
	Therefore, consider any $\boundpair$-partition $\costr$. 
	Let $N'_1$ be a coalition from $\costr$ with $|N'_1|\neq \lb$.
	Such a coalition exists as $\lb$ does not divide $n$.
	As argued above, there must exist an agent $\genA\in N'_1$ such that there exists $\genB \not \in N'_1$ with $v_{\genA}(\genB) = 1$. 
	If $|\costr (\genB)| < \ub$, then we have found a $\boundpair$-feasible IS\feastab deviation ($\genA$ deviates from $\costr (\genA)$ to $\costr (\genB)$). Otherwise we set $N'_2 = N'_1\cup \costr (\genB)$.
	
	We now iterate this process.
	Whenever we are given a union of coalitions $N'_k$, we check whether $N'_k = N$.
	If this is not the case, we find an agent $\genA\in N'_k$ such that there exists $\genB \not \in N'_k$ with $v_{\genA}(\genB) = 1$.
	If $|\costr (\genB)| < \ub$, we find a $\boundpair$-feasible IS\feastab deviation where $\genA$ joins $\costr(\genB)$.
	Note that this leaves the abandoned coalition of size at least $\lb$ because $N'_1$ was larger than $\lb$ and all other coalitions in the union are of size $\ub$.
	If $|\costr (\genB)| = \ub$, we update $N'_{k+1} = N'_k\cup \costr(\genB)$.
	
	This process can only end if $N'_k = N$. However, this implies that $\costr$ consists only of coalitions of size $\ub$, except possibly $N'_1$.
	Since $\ub$ does not divide $n$, we must have that $|N'_1| < \ub$.
	But then we find an agent $\genA\notin N'_1$ whose unique successor in the cycle is in $N'_1$.
	Hence, $\genA$ can perform a $\boundpair$-feasible IS\feastab deviation to join $N'_1$, abandoning a coalition of size $\ub$.
\end{proof}

Similarly, without symmetry, we lose existence of CNS\feastab.
Therefore, we consider instances where most valuations are zero, except for three agents forming a directed cycle with negative valuations and a partitioning of the other agents into pairs with mutual negative valuations.

\begin{restatable}{proposition}{nonexistCNSasy}\label{prop:nonexistCNSasy}
    Let $2\le \lb < \ub$. 
    Then there exists an ASHG
    that contains no CNS\feastab $(\lb,\ub)$-partition.
\end{restatable}

\begin{proof}
	Consider any $2\le \lb < \ub$. 
	We define an ASHG $(N,\vf)$ where $N = \{\genA_1,\dots,\genA_{\lb-1}, \genB_1,\dots,\genB_{\lb-1},c_1,c_2,c_3\}$. 
	Valuations assume only the values in $\{-1,0\}$ and are given as follows. 
	\begin{itemize}
		\item We have $\vf_{\genA_{i}}(\genB_{i}) = \vf_{\genB_{i}}(\genA_{i}) = -1$ for all $i \in [\lb-1]$.
		\item We have $\vf_{c_1}(c_2) = \vf_{c_2}(c_3) = \vf_{c_3}(c_1) = -1$.
		\item All other valuations are~$0$. 
	\end{itemize}
	
	Hence, the game consists of a directed negative cycle of three agents and further mutually negative pairs of agents.
	An illustration is provided in \Cref{fig:nonexistCNSstar}.
	
	\begin{figure}[tb]
		\centering
		\begin{tikzpicture}[every node/.style={draw, circle, minimum size=.8cm, inner sep=0pt}]
						\foreach[count = \k] \i/\j in {0/1,1.3/2}
			{
				\node (a\k) at (\i,0) {$\genA_{\j}$};
				\node (b\k) at (\i,1.8) { $\genB_{\j}$};
			}
			\foreach \i/\j/\k in {3.2/{\lb - 1}/3}
			{
				\node (a\k) at (\i,0) {\footnotesize $\genA_{\j}$};
				\node (b\k) at (\i,1.8) {\footnotesize $\genB_{\j}$};
			}
			\foreach \i in {1,2,3}{
				\draw[->, thick, myred,bend right] (a\i) edge (b\i);
				\draw[->, thick, myred,bend right] (b\i) edge (a\i);
			}
			\node[draw=none] at ($(a2)!.5!(b3)$) {\dots};
			
						\node (c1) at ($(-2.5,.9) + (60:1)$) {$c_1$};
			\node (c2) at ($(-2.5,.9) + (180:1)$) {$c_2$};
			\node (c3) at ($(-2.5,.9) + (300:1)$) {$c_3$};
			
			\foreach \i/\j in {1/2,2/3,3/1}{
				\draw[->, thick, myred,bend right] (c\i) edge (c\j);}
		\end{tikzpicture}
		\caption{ASHG constructed in the proof of \Cref{prop:nonexistCNSasy}.
			Red edges indicate valuations of~$-1$, while missing edges indicate a valuation of~$0$.}
		\label{fig:nonexistCNSstar}
	\end{figure}
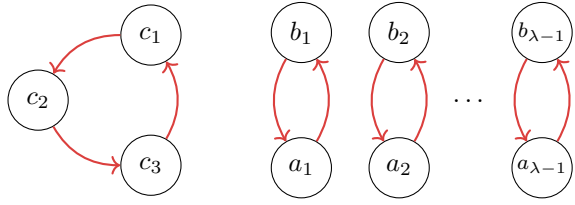
	
							Let $\costr$ be any $\boundpair$-partition.
	By design of the number of agents, $\costr$ has precisely $2$ coalitions. 
	One, which we call $S_1$, has $\lb+1$ agents, the other, which we call $S_2$, has $\lb$ agents.
	
	For $\costr$ to be a CNS$^*$ partition, it must be the case that $\genA_{i}$ and $\genB_{i}$ are not both in $S_1$ for all $i \in [\lb-1]$. 
	Indeed, if $\genA_{i} \in S_1$ and $\genB_{i} \in S_1$ then $\genA_{i}$ could make a CNS$^*$ deviation to $S_2$. 
	
	Thus at most $\lb-1$ agents in $S_1$ are from $\{\genA_1,\dots,\genA_{\lb-1}, \genB_1,\dots,\genB_{\lb-1}\}$. This means at least $2$ of the agents in $S_1$ are from $\{c_1,c_2,c_3\}$. 
	Without loss of generality (by symmetry), we assume $c_1,c_2 \in S_1$. 
	But then $c_2$ can make a CNS$^*$ deviation to $S_2$, as no agent has a positive preference value for $c_2$ (as all preference values are nonpositive), and in making this deviation $c_2$ improves their utility from $-1$ to $0$. 
	Moreover, by the sizes of the coalitions, this is a $\boundpair$-feasible deviation.
	We conclude that the constructed ASHGs admit no CNS$^*$ partitions.
\end{proof}

Notably, the nonexistence results of \Cref{prop:nonexistCISsymm,prop:nonexistISasy,prop:nonexistCNSasy} assume a lower bound of $\lb \ge 2$. 
However, in the next section, we will obtain a complete complexity picture of all stability notions when $\lb = 1$ and thereby also exactly characterize the existence of stability in this case.

\section{Computation of Stable Partitions}

In this section, we consider the computational complexity of deciding whether there exist stable $(\lb,\ub)$-partitions in a given ASHG.
We obtain positive results in the form of efficient algorithms as well as \NP-completeness results for the decision problem of whether a stable $(\lb,\ub)$-partition exists.
Note that for all of our stability notions, membership in \NP{} is straightforward: a stable partition can easily be verified in polynomial time, by considering all possible deviations from all agents to all coalitions (possibly including the empty coalition if $\lb = 1$).
Hence, we will omit this from all hardness proofs.

In the first part of the section, we will consider settings when $\lb = 1$, i.e., partitions only have to satisfy upper bounds on their size. 
We will obtain a complete picture for the computational complexity for all possible parameters $\ub$.
In the second part of the section, we provide results for a nontrivial lower bound of $\lb \ge 2$.

\subsection{Upper-Bounded Coalition Sizes}\label{sec:upperboundedcomplexity}

We start with the consideration of settings where $\lb = 1$.
Note that, under this restriction, permissible and feasible deviations are the same because abandoning a coalition cannot result in an infeasible coalition.
Hence, all results for this section hold for both standard and their feasible stability concepts.
For simplicity, we state them all for standard stability concepts.
In addition, a $(1,\ub)$-partition of any number of agents always exists as, for example, the singleton partition is such a partition.
Hence, we are not concerned with the nonexistence of feasible partitions in this section.

We start with the consideration of CIS.
There, we obtain a polynomial-time algorithm for any fixed upper bound.
We remark that \citet{ABS11c} present an algorithm to compute CIS partitions in ASHG without size bounds.
Their algorithm works by iteratively choosing a leader among agents not assigned to coalitions yet.
The idea then is to let the leader either form a new coalition with all their friends or, if this yields a higher utility, joining an existing coalition.
When joining an existing coalition, further agents, so called latecomers, are added if this is beneficial for members in this coalition while not harming any agent there.
Unfortunately, their algorithm is not correct.
The high-level reason is that, when joining a coalition, an agent might prefer joining a different coalition because of the utility gain through latecomers. 
In \Cref{app:CISflawed}, we present a concrete example where the algorithm does not return a CIS partition.
Our algorithm fixes their construction by changing the condition of joining coalitions to immediately take into consideration the gains of other agents that can join as well.
We additionally modify it to work for an arbitrary upper bound on coalition sizes.
However, note that the algorithm can also be used to compute CIS partitions for ASHGs without size constraints by running it for $\ub = n$.

\begin{theorem}\label{thm:1ubCIS}
    Let $\ub \ge 2$.
    Then there exists a polynomial-time algorithm that computes CIS $(1,\ub)$-partitions in ASHGs.
\end{theorem}

\begin{algorithm}[tb]
	\caption{Contractually individually stable $(1,\ub)$-partitions in ASHGs.}\label{algorithm:CISub}
    \begin{flushleft}
	\textbf{Input:} ASHG $(N,\vf)$, parameter $\ub$\\
	\textbf{Output:} $(1,\ub)$-partition
    \end{flushleft}
	\begin{algorithmic}[1]
	\State $i\leftarrow 0$, $A\leftarrow N$ \Comment{Initialize coalition counter and still available agents}
			  \While{$A\neq \emptyset$}\label{while-step}
	\State Select  any agent $\genA\in A$ 
    	\State $h \leftarrow \sum_{\genB\in \topset{\ub-1}{\text{Fr}(\genA,A)}}\vf_{\genA}(\genB)$ \Comment{Utility $\genA$ gets from creating their best new coalition}
	\State $z\leftarrow i+1$
        \Comment{Indicating coalition chosen by $a$. 
        }
	\For{$k \in [i]$} \Comment{Skipped for $i = 0$ as $[0] = \emptyset$}
    \State $P_k \leftarrow \text{Fr}(\genA,A) \cap \left(A \setminus \text{En}(\widetilde{S}_k,A) \right)$ 
    \State\Comment{Approved helpers when joining $S_z$}
	\State $h'\leftarrow \sum_{\genB\in S_k}\vf_{\genA}(\genB) + \sum_{\genB\in \topset{\ub-|S_k|-1}{P_k}}\vf_{\genA}(\genB) $ \State\Comment{Utility when joining $S_z$ with best helpers}
	\If{ $h < h'$ 
    }
	\State $h\leftarrow h'$, $z \leftarrow k$
	\EndIf
	\EndFor

	\If{$z = i+1$} \Comment{$\genA$ is leader}

	\State $S_z\leftarrow \{\genA\}\cup \topset{\ub-1}{\text{Fr}(\genA,A)}$ \Comment{Create best new coalition}
        \State $\widetilde{S}_z\leftarrow \{\genA\}$ \Comment{Add leader as coalition's decision maker}
	\State $A\leftarrow A \setminus S_z$
	\State $i\leftarrow i + 1$
    \Else \Comment{$\genA$ is latecomer.}
	\State $S_{z}\leftarrow S_{z}\cup \{\genA\}\cup \topset{\ub-|S_z|-1}{P_z}$ \Comment{Add latecomer and their helpers to the coalition}
        \State $\widetilde{S}_{z}\leftarrow \widetilde{S}_{z}\cup\{\genA\}$ \Comment{Add latecomer to coalition's decision makers}
	\State $A\leftarrow A \setminus (\{\genA\}\cup \topset{\ub-|S_z|-1}{P_z})$

	 \EndIf
		\EndWhile\\
        
	  \Return $\costr = \{S_1,\ldots, S_i\}$
	 \end{algorithmic}

\end{algorithm}

\begin{proof}
    Consider an ASHG $(N,\vf)$.
    We apply \Cref{algorithm:CISub}.
    Note that the algorithm uses notation for agent sets as defined in \Cref{sec:agentsets}.
    The algorithm initializes a counter $i$ capturing the number of already created coalitions.
    Moreover, it keeps track of a set $A$ of still available agents. 
    The idea is to iteratively create coalitions $S_i$ based on the preferences of certain \emph{decision makers} which are stored in a set $\widetilde{S}_{i}$.
    Decision makers add as many friends as the capacity of a coalition allows to that coalition.
    Intuitively, the resulting partition is a CIS partition because decision makers always obtain a best coalition among the available partitions whereas agents that are not decision makers are denied to leave their coalition by a decision maker whose friend they are.
    
    We start by selecting an available agent from $\genA \in A$, and call them the \emph{leader} of the first coalition $S_1$. 
    We try to form a coalition that is as good as possible for the leader.
    For this, we consider the up to $\ub - 1$ best agents among the friends of $a$ among the available agents, i.e., $\topset{\ub - 1}{\text{Fr}(\genA,A)}$.
    As a first coalition, we form $S_1 = \{\genA\}\cup \topset{\ub - 1}{\text{Fr}(\genA,A)}$ and indicate that $\genA$ is a decision maker of this coalition by setting $\widetilde{S}_{z} = \{\genA\}$.
    We call the agents enforced to join the coalition by a decision maker the \textit{helpers}. 
    
    In subsequent iterations, we start by selecting a new leader $\genA\in A$. 
    However, apart from creating a new coalition, $\genA$ may also join an existing coalition, while bringing along friends, provided that they do not violate the upper bound and that both they and their friends are not enemies of a decision maker in the joined coalition. 
    Under these conditions, $\genA$ either creates a new coalition or joins an existing coalition to maximize their utility.
    If $\genA$ joins an existing coalition, we refer to $\genA$ as a \emph{latecomer} and add them as a decision maker. 
    Additionally, we refer to the friends they bring along as \emph{helpers}, analogous to the helpers of the first leader.
    Finally, we update the set of available agents by removing $\genA$ and $\genA$'s helpers.
    We continue selecting new agents to become leaders or latecomers until the set $A$ of remaining agents is empty.

    Clearly, the algorithm runs in polynomial time.
    Each agent is processed at most once as a leader and we have to compute at most $n = |N|$ possible options for them.
    Finding the best $k$ agents among a set and performing the updates of the coalition structure are all straightforward.

    We now show that the algorithm produces a CIS $(1,\ub)$-partition $\costr$.
    Clearly, no created coalition contains more than $\ub$ agents, so we create a $(1,\ub)$-partition.
    It remains to show that no agent can perform a $(1,\ub)$-permissible CIS deviation.
    For this, we consider potential deviating agents case by case, and show that no agents can make such a deviation.

    First, for every helper $\genB$, there exists a decision maker $\genA$ with $\genB\in \text{Fr}(\genA,N)$, i.e., $\vf_{\genA}(\genB) > 0$.
    Hence, $\genB$ is denied by $\genA$ to perform a CIS deviation.

    Now, consider a decision maker $\genA$.
    Consider the iteration of the while loop during which $\genA$ was the decision maker and the set of available agents was $A$.
    We refer to the partial partition of the subset of agents that were assigned to a coalition after this iteration as $\costr'$.
    Note that, after $\genA$ formed a new coalition or joined a coalition, no agent could have joined this coalition that was an enemy of $\genA$. 
    Hence, we have that 
    \begin{equation}\label{eq:compwithres}
        \uf_{\genA}(\costr)\ge \uf_{\genA}(\costr')\text.
    \end{equation}
            In particular, this implies that $a$ is individually rational, and thus does not wish to deviate to a singleton coalition. 
    
    It remains to show that $a$ cannot perform a CIS deviation into an existing coalition.
    Let $C\in \costr\setminus \{\costr(\genA)\}$ be some coalition, and define $C' = \{\genB\in C\colon \exists D\in \costr' \text{ with }\genB \in D\}$.
    Thus, $C'$ is the subset of agents of $C$ that were in a coalition after the iteration with leader $\genA$.
    Note that $C'$ is either a nonempty coalition of $\costr'$ or it is the empty set, indicating that $C$ was only created after this iteration.
    
    We consider two cases. 
    First, assume that $C' = \emptyset$.
    Since~$i$ had the opportunity to create     $\{\genA\}\cup \topset{\ub - 1}{\text{Fr}(\genA,A)}$, they have a higher utility than for this coalition. 
    Hence,
    $\uf_{\genA}(\costr')\ge \uf_{\genA}(\{\genA\}\cup \topset{\ub - 1}{\text{Fr}(\genA,A)}) \ge \uf_{\genA}(\{\genA\}\cup C\})$.

    Second, assume $C' \neq \emptyset$. Let $\widetilde S\subseteq C'$ be the set of decision makers in $C'$. 
    If $\genA \in \text{En}(\widetilde{S},N)$, then $\genA$ is denied to enter $C'$ and therefore $C$. Thus $\genA$ cannot make a CIS deviation to $C$.
    Hence, assume that $\genA \notin \text{En}(\widetilde{S},N)$
    
    Note that $C \subseteq C'\cup(A\setminus \text{En}(\widetilde{S},A))$, as by design $C$ cannot contain enemies of its own decision makers.
    Additionally, $\genA$ had the opportunity to form the coalition $D = C' \cup \{\genA\}\cup \topset{\ub - |C'|-1}{P}$ where $P = \text{Fr}(\genA,A) \cap \left(A \setminus \text{En}(\widetilde{S},A) \right)$.
            Hence, $\uf_{\genA}(\costr')\ge \uf_{\genA}(D) \ge \uf_{\genA}(\{\genA\}\cup C\})$.

    Thus, whenever $\genA$ is not denied entry to $C$ by a decision maker of $C'$, we have obtained $\uf_{\genA}(\costr')\ge \uf_{\genA}(\{\genA\}\cup C\})$.
    Combining this with \Cref{eq:compwithres}, we obtain that $\uf_{\genA}(\costr)\ge \uf_{\genA}(\{\genA\}\cup C\})$.
    Hence, $\genA$ cannot perform a     CIS deviation.
\end{proof} 

Our next two results concern the computation of CNS partitions.
We first obtain an efficient algorithm when $\ub = 2$.

\begin{restatable}{proposition}{onetwoCNS}\label{prop:12CNS}
    There exists a polynomial-time algorithm that computes CNS $(1,2)$-partitions in ASHGs.
\end{restatable}

\begin{algorithm}[tb]
	\caption{Contractually Nash-stable $(1,2)$-partitions in ASHGs.}\label{algorithm:CNSmatching}
    \begin{flushleft}
    	\textbf{Input:} ASHG $(N,\vf)$ with $N = \{\genA_1,\dots, \genA_n\}$\\
	   \textbf{Output:} $(1,2)$-partition
    \end{flushleft}
	\begin{algorithmic}[1]
		\State{$\costr \longleftarrow \emptyset$, $A \longleftarrow N$}
        \Comment{Empty partition and available agents}
        		        		
		\For{$k \in [n]$}
        \If {$\genA_k\in A$ and $\{\genA\in A\colon \vf_{\genA_k}(\genA) > 0\}\neq \emptyset$}
        \State{Choose $i^*\in \arg\max_{\genA\in A}\{\vf_{\genA_k}(\genA)\}$} 
                                \State{$\costr \longleftarrow \costr\cup\{\{\genA_k,i^*\}\}$, $A\longleftarrow A\setminus \{\genA_k,i^*\}$}
		\EndIf
		\EndFor
		\State{$\costr \longleftarrow \costr \cup \{\{\genA\}\colon \genA\in A\}$}
		\State{\Return{$\costr$}} 
	\end{algorithmic}
\end{algorithm}

\begin{proof}
    Consider an ASHG $(N,\vf)$ and with agent set $N = \{\genA_1,\dots, \genA_n\}$.
    We apply \Cref{algorithm:CNSmatching}. 
    The algorithm initializes an empty coalition structure and keeps track of a set $A$ of available agents that initially contains all agents.
    Then, agents are considered in a fixed order to decide whether to put them in a coalition of size~$2$ or whether they should stay alone.
    At their turn, if an agent is still available and can form a coalition with another available agent achieving a positive utility, then they are assigned to the best such coalition, and both agents are removed from the set of available agents.

    Clearly, the algorithm runs in polynomial time.
    Moreover, we claim that the produced partition is contractually Nash-stable.
    First, consider an agent $\genA\in N$ that is assigned to a singleton coalition.
    Then, as $\ub = 2$, the only permissible deviations are to join other singleton coalitions.
    However, such a coalition cannot yield a positive utility for the agent as all agents in singleton coalitions were available to form a coalition when it was $\genA$'s turn in the execution of the for loop.

    Second, consider an agent $\genA$ that is assigned to a coalition of size~$2$.
    If this coalition was formed during $\genA$'s turn, then the formed coalition yields the highest utility among coalitions of size~$2$ with available agents.
    Hence, deviating to a singleton coalition created at the end of the algorithm cannot improve $\genA$'s utility.
    Finally, assume that $\genA$ was assigned to a coalition during the turn of another agent $\genB$.
    Then, $\vf_{\genB}(\genA) > 0$ and $\genA$ is vetoed from performing a CNS deviation.
    We conclude that the produced partition is contractually Nash-stable.
\end{proof}

Unfortunately, as we show next, \Cref{prop:12CNS} does not extend to larger upper bounds.
Our proof is inspired by the existing hardness result for the existence of CNS partitions in ASHGs with unconstrained coalition sizes \cite[Theorem~3]{BBT23a}.
However, we apply some crucial modifications to avoid the formation of coalitions of size larger than~$3$.
We focus on presenting the general idea of the proof and defer a full proof to \Cref{app:upperbound}.

\begin{restatable}{theorem}{CNSugeThree}\label{thm:CNS-Uge3}
    Let $\ub\ge 3$.
    Then it is \NP-complete to decide whether a CNS $(1,\ub)$-partition exists in ASHGs.
\end{restatable}

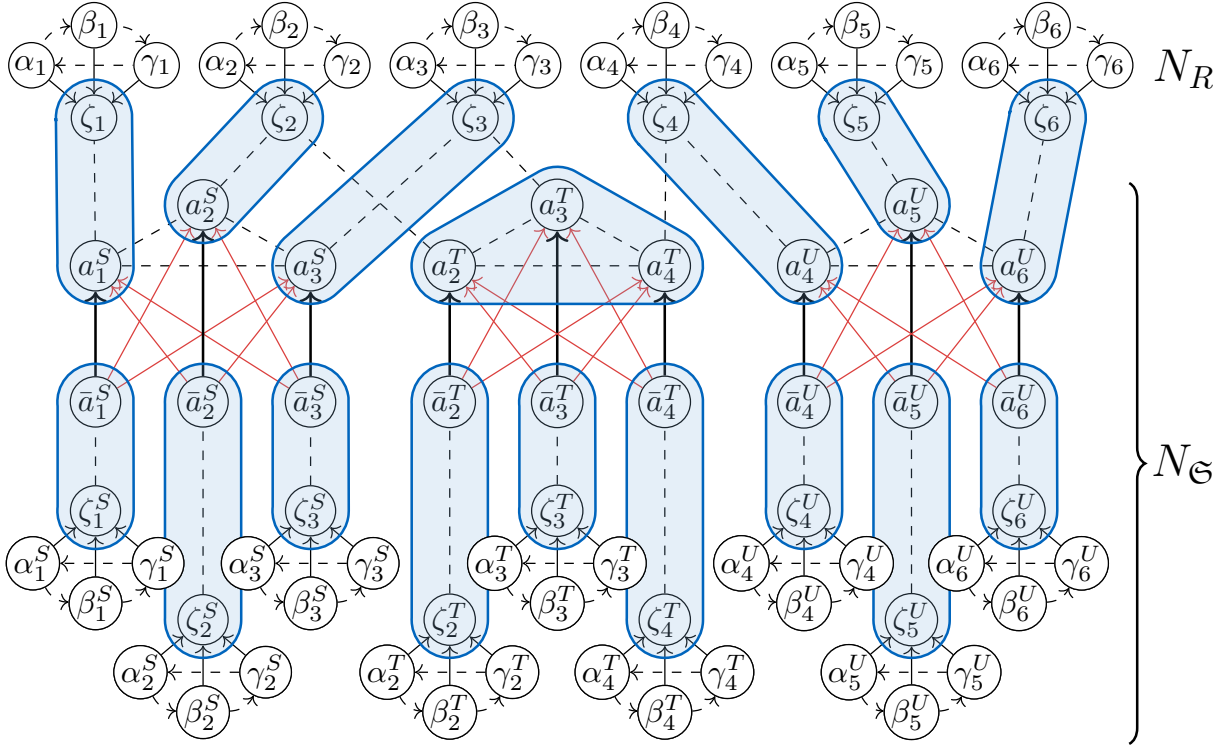
\begin{figure*}[tb]
	\centering
	\resizebox{1\textwidth}{!}{
	\begin{tikzpicture}[scale = .6, element/.style={shape=circle,draw, fill=white, inner sep=0.1, minimum size=.5cm}]
    \pgfmathsetmacro\Sstretch{2.1} 	\foreach[count = \k] \x in {1,...,6}{
		\node[element] (c\x) at (3.5*\x,-1.8) {${\zeta}_{\k}$};
		\node[element] (t\x) at ($(c\x) + (90:1.7)$) {${\grk{2}}_{\k}$};
		\node[element] (l\x) at ($(c\x) + (140:1.5)$) {${\grk{1}}_{\k}$};
		\node[element] (r\x) at ($(c\x) + (40:1.5)$) {${\grk{3}}_{\k}$};
		\draw[->] (t\x) edge (c\x);
		\draw[->] (l\x) edge (c\x);
		\draw[->] (r\x) edge (c\x);
		\draw[<-, bend right, dashed] (t\x) edge (l\x);
		\draw[<-, dashed] (l\x) edge (r\x);
		\draw[<-, bend right, dashed] (r\x) edge (t\x);
	}

	\node (s1) at (5.5,-3.8) {};
	\node (s2) at (12,-3.8) {};
	\node (s3) at (18.5,-3.8) {};
		\node[element] (s13) at ($(s1) + (340:\Sstretch)$) {$a^S_3$};
		\node[element] (s11) at ($(s1) + (200:\Sstretch)$) {$a^S_1$};
		\node[element] (s12) at ($(s1) + (90:0.4)$) {$a^S_2$};
		\node[element] (sb13) at ($(s13) + (0,-2.5)$) {$\bar a^S_3$};
		\node[element] (sb11) at ($(s11) + (0,-2.5)$) {$\bar a^S_1$};
		\node[element] (sb12) at ($(s13)!.5!(s11) + (0,-2.5)$) {$\bar a^S_2$};
		\node[element] (s23) at ($(s2) + (340:\Sstretch)$) {$a^T_4$};
		\node[element] (s21) at ($(s2) + (200:\Sstretch)$) {$a^T_2$};
		\node[element] (s22) at ($(s2) + (90:0.4)$) {$a^T_3$};
		\node[element] (sb23) at ($(s23) + (0,-2.5)$) {$\bar a^T_4$};
		\node[element] (sb21) at ($(s21) + (0,-2.5)$) {$\bar a^T_2$};
		\node[element] (sb22) at ($(s21)!.5!(s23) + (0,-2.5)$) {$\bar a^T_3$};
		\node[element] (s33) at ($(s3) + (340:\Sstretch)$) {$a^U_6$};
		\node[element] (s31) at ($(s3) + (200:\Sstretch)$) {$a^U_4$};
		\node[element] (s32) at ($(s3) + (90:0.4)$) {$a^U_5$};
		\node[element] (sb33) at ($(s33) + (0,-2.5)$) {$\bar a^U_6$};
		\node[element] (sb31) at ($(s31) + (0,-2.5)$) {$\bar a^U_4$};
		\node[element] (sb32) at ($(s31)!.5!(s33) + (0,-2.5)$) {$\bar a^U_5$};
	\foreach \s in {1,2,3}{
		\draw[dashed] (s\s1) edge (s\s2);
		\draw[dashed] (s\s2) edge (s\s3);
		\draw[dashed] (s\s3) edge (s\s1);
		\foreach \t in {1,2,3}{
        \draw[->,thick] (sb\t\s) edge (s\t\s);
   			            			}
    \draw[->,myred] (sb\s1) edge (s\s3);
    \draw[->,myred] (sb\s1) edge (s\s2);
    \draw[->,myred] (sb\s2) edge (s\s3);
    \draw[->,myred] (sb\s2) edge (s\s1);
    \draw[->,myred] (sb\s3) edge (s\s1);
    \draw[->,myred] (sb\s3) edge (s\s2);
    }

	\draw[dashed] (s11) edge (c1);
	\draw[dashed] (s12) edge (c2);
	\draw[dashed] (s13) edge (c3);
	\draw[dashed] (s21) edge (c2);
	\draw[dashed] (s22) edge (c3);
	\draw[dashed] (s23) edge (c4);
	\draw[dashed] (s31) edge (c4);
	\draw[dashed] (s32) edge (c5);
	\draw[dashed] (s33) edge (c6);

    \node[element] (S01) at ($(sb11) + (0,-2)$) {$\zeta_1^S$};
    \node[element] (S02) at ($(sb12) + (0,-4)$) {$\zeta_2^S$};
    \node[element] (S03) at ($(sb13) + (0,-2)$) {$\zeta_3^S$};
	
    \node[element] (T01) at ($(sb21) + (0,-4)$) {$\zeta_2^T$};
    \node[element] (T02) at ($(sb22) + (0,-2)$) {$\zeta_3^T$};
    \node[element] (T03) at ($(sb23) + (0,-4)$) {$\zeta_4^T$};
	
    \node[element] (U01) at ($(sb31) + (0,-2)$) {$\zeta_4^U$};
    \node[element] (U02) at ($(sb32) + (0,-4)$) {$\zeta_5^U$};
    \node[element] (U03) at ($(sb33) + (0,-2)$) {$\zeta_6^U$};

		    \foreach \i in {1,2,3}{
    \foreach \j/\k/\la/\lb/\lc in {S/1/1/2/3,T/2/2/3/4,U/3/4/5/6}{
	\node[element] (t0\i) at ($(\j0\i) + (270:1.7)$) {${\grk{2}}_{\lb}^{\j}$};
	\node[element] (l0\i) at ($(\j0\i) + (320:1.5)$) {${\grk{3}}_{\lc}^{\j}$};
	\node[element] (r0\i) at ($(\j0\i) + (220:1.5)$) {${\grk{1}}_{\la}^{\j}$};
	\draw[->] (t0\i) edge (\j0\i);
	\draw[->] (l0\i) edge (\j0\i);
	\draw[->] (r0\i) edge (\j0\i);
	\draw[->, bend right, dashed] (t0\i) edge (l0\i);
	\draw[->, dashed] (l0\i) edge (r0\i);
	\draw[->, bend right, dashed] (r0\i) edge (t0\i);
    \draw[dashed] (\j0\i) edge (sb\k\i);
    \draw[thick,myblue, fill=myblue!50, fill opacity=0.2] \convexpath{\j0\i, sb\k\i}{0.7cm};
					    }     } 
    
    \foreach \i/\j in {1/1,2/2,3/3}{
	\node[element] (tS\i) at ($(S0\i) + (270:1.7)$) {${\grk{2}}_{\j}^{S}$};
	\node[element] (lS\i) at ($(S0\i) + (320:1.5)$) {${\grk{3}}_{\j}^{S}$};
	\node[element] (rS\i) at ($(S0\i) + (220:1.5)$) {${\grk{1}}_{\j}^{S}$};
    }     \foreach \i/\j in {1/2,2/3,3/4}{
	\node[element] (tT\i) at ($(T0\i) + (270:1.7)$) {${\grk{2}}_{\j}^{T}$};
	\node[element] (lT\i) at ($(T0\i) + (320:1.5)$) {${\grk{3}}_{\j}^{T}$};
	\node[element] (rT\i) at ($(T0\i) + (220:1.5)$) {${\grk{1}}_{\j}^{T}$};
    }     \foreach \i/\j in {1/4,2/5,3/6}{
	\node[element] (tU\i) at ($(U0\i) + (270:1.7)$) {${\grk{2}}_{\j}^{U}$};
	\node[element] (lU\i) at ($(U0\i) + (320:1.5)$) {${\grk{3}}_{\j}^{U}$};
	\node[element] (rU\i) at ($(U0\i) + (220:1.5)$) {${\grk{1}}_{\j}^{U}$};
    } 
            \draw[thick,myblue, fill=myblue!50, fill opacity=0.2] \convexpath{c1, s11}{0.7cm};
    \draw[thick,myblue, fill=myblue!50, fill opacity=0.2] \convexpath{c2, s12}{0.7cm};
    \draw[thick,myblue, fill=myblue!50, fill opacity=0.2] \convexpath{c3, s13}{0.7cm};
    \draw[thick,myblue, fill=myblue!50, fill opacity=0.2] \convexpath{c4, s31}{0.7cm};
    \draw[thick,myblue, fill=myblue!50, fill opacity=0.2] \convexpath{c5, s32}{0.7cm};
    \draw[thick,myblue, fill=myblue!50, fill opacity=0.2] \convexpath{c6, s33}{0.7cm};
        \draw[thick,myblue, fill=myblue!50, fill opacity=0.2] \convexpath{s21, s22, s23}{0.7cm};

        \draw [thick,decorate,decoration={brace,amplitude=5pt}]
		($(22.5,-3)$) -- ($(22.5,-13.3)$) node[midway, xshift = .6cm] (NS) {\Large $N_{\etcS{}}$};
		\node at ($(NS)+(0,7.25)$) {\Large $N_R$}; 
	\end{tikzpicture}
	} 
	\caption{Schematic of the reduction from the proof of \Cref{thm:CNS-Uge3}. We depict the reduced instance for the instance $\langle R,\etcS{}\rangle$ of \etcs{} where $R = \{1,2,3,4,5,6\}$, and $\etcS{} = \{S,T,U\}$, with $S = \{1,2,3\}$, $T = \{2,3,4\}$, and $U = \{4,5,6\}$.     Exact valuations are defined in the appendix.
		Thin edges represent a valuation of~$1$, thick edges represent a valuation of~$2$, and red edges represent a valuation of $-1$. 
        Dashed edges represent a valuation of~$0$. 
        If the edge is undirected, the same valuation applies in both directions.
        Omitted edges represent a negative utility of $-3$.
        In blue, we highlight the nonsingleton coalitions of the CNS partition corresponding to the cover $\{S,U\}$.}\label{fig:CNS-general}
\end{figure*}

\begin{proof}[Proof idea]
    We describe our proof along the visualization given in \Cref{fig:CNS-general}.
    The reduction is from the \NP-complete problem \etc{} (\etcs) \cite{GaJo79a}.
    An instance of \etcs{} is a pair $\langle R,\etcS{}\rangle$, where $R$ is a ground set of size $3\rho$ and $\etcS{}$ is a collection of $3$-element subsets of $R$; it is a Yes-instance if and only if there exists a subset $\etcS{}'\subseteq \etcS{}$ that partitions~$R$, i.e., covers $R$ and satisfies $|{\etcS{}}'|=\rho$.

    Given an instance $\langle R,\etcS\rangle$, we construct a reduced instance, where every element $r\in R$ is represented by a gadget with agent set $\{\zeta_r,\alpha_r,\beta_r,\gamma_r\}$ 
        corresponding to a known ASHG lacking a CNS partition \cite{SuDi07b}.
    This property is even true for any upper bound $\ub \ge 2$.
    Hence, some agent in this gadget has to be in a coalition with an agent outside the gadget.

    Now, every set $S\in \etcS$ is represented by~$6$ agents: for every element $r\in S$, there are agents $a_r^S$ and $\bar a_r^S$.
    Agents of the type $a_r^S$ can prevent the instability of the gadget associated with $r$, while agents $\bar a_r^S$ want to be with $a_r^S$.
    However, agents $\bar a_r^S$ each also have an associated gadget of the aforementioned $4$-player game without a CNS partition. Stability thus relies on $\bar a_r^S$ cooperating with this gadget, and so there should not be an incentive for them to join $a_r^S$.
    This is the case if and only if the three $a_r^S$ agents representing $S$ either form a coalition together (like for the set $T$ in \Cref{fig:CNS-general}) or are in coalitions with the gadgets representing their corresponding agents in $R$ (like for the sets $S$ and $U$ in \Cref{fig:CNS-general}).
    In this way, a CNS partition exists if and only if $R$ can be partitioned by sets in $\etcS$.
\end{proof}

It is known that it is \NP-complete to determine whether there exists a Nash-stable or individually stable matching, which are the same as $(1,2)$-partitions \cite{Aziz13a}.
A similar approach works for any other upper bound.
Interestingly, our reduction allows for a simultaneous treatment of Nash and individual stability. 
As \citet{Aziz13a}, we reduce from the \NP-complete problem \mmm{} \cite{HoKi93a}. The proof is presented in \Cref{app:upperbound}.

\begin{restatable}{theorem}{NSandIS}\label{thm:NSandIS1u}
    Let $\ub\ge 2$.
    Then it is \NP-complete to decide whether an NS (or IS) $(1,\ub)$-partition exists in ASHGs.
\end{restatable}

\subsection{Coalitions with Nontrivial Lower Bound}\label{sec:nontrivialLB}

We now turn to the consideration of finding stable $\boundpair$-partitions, when $\lb\ge 2$.
First, we show a negative result for NS for any pair $(\lb,\ub)$ with $\ub\ge 4$.
We discuss a proof idea here and defer the full proof to \Cref{app:bothbounds}.

\begin{restatable}{theorem}{NSugefour}\label{thm:NSuge4}
    Let $\ub\ge 4$ and $\lb < \ub$.
    Then it is \NP-complete to decide whether an NS $(\lb,\ub)$-partition exists in an ASHG.
\end{restatable}

\begin{proof}[Proof idea]
    We provide another reduction from \etcs{}.
    Given an instance $(R,\etcS)$, we create a reduced ASHG as follows: 
    each element in $R$ is represented by a single agent, while each set in $\etcS$ is represented by $\ub-3$ agents.
    This should incentivize the creation of coalitions of all agents representing a set in $S\in \etcS$ together with the three agents representing the elements from $R$ contained in $S$.
    To give other set-representing agents a good alternative, we have further triplets of agents that can fill their coalitions.     This way, we want to ensure a correspondence with Yes-instances of \etcs.
    
    However, agents might have incentives to form other coalitions.
        To control these, we introduce a special agent $\alpha$ that plays a ``cat-and-mouse'' dynamics with all the aforementioned agents, to which we refer as core agents.
    While $\alpha$ would like to form a coalition with core agents, they want to avoid $\alpha$.
    In particular, this ensures that coalitions of core agents always have to be full, and, with a bit more effort, of the aforementioned form.
    Finally, if $\lb\ge 2$, the ``cat-and-mouse'' dynamics faces the difficulty that agents cannot simply form a new coalition when their coalition is not individually rational.
    For this, we introduce a sufficiently large set of dummy agents, such that there is always a nonfull coalition to which a deviation is possible.
\end{proof}

The combination of \Cref{thm:NSandIS1u,thm:NSuge4} implies that computing NS $(\lb,\ub)$-partitions is \NP-hard for all combinations of $\lb$ and $\ub$ with $\lb < \ub$ except for $\lb = 2$ and $\ub =3$, where we leave the complexity open.

Finally, we work towards achieving CIS\feastab for a nontrivial lower bound on coalition sizes $\lb \ge 2$.
As a first result, we present an algorithm for ASHGs $(N,\vf)$ with nonzero valuations, i.e., $\vf_i(j)\neq 0$ for all $i,j\in N$.
Such games are sometimes referred to as strict ASHGs \cite{ABS11c}.
The idea is similar as for \Cref{algorithm:CISub}: we iteratively choose leaders who form their best possible coalitions.
However, we have to form coalitions of size at least $\lb$ and, therefore, the formed coalitions might have to contain agents, for which the leader has a negative valuation.
Still, we assign them their best coalition of size $\lb$.
If the leader has even more friends, we add them up to the constraints given by the upper bound and by the number of agents that may be added to coalitions of size larger than $\lb$.
In this way, we create $k$ coalitions for a $k$ that allows for a $\boundpair$-partition into $k$ coalitions.
The remaining agents are added to coalitions in inverse order arbitrarily.
This order ensures that only later-assigned leaders may be forced to accommodate enemies, which cannot join coalitions of earlier-assigned leaders.

\begin{restatable}{theorem}{nonzeroCIS}\label{thm:nonzeroCIS}
	For ASHGs with nonzero valuations, there exists a polynomial-time algorithm that computes a CIS\feastab $(\lb,\ub)$-partition into $k$ coalitions, or determines that no such partition exists.
\end{restatable}

\begin{algorithm}[tb]
	\caption{CIS\feastab $\boundpair$-partitions in ASHGs with nonzero valuations.}\label{algorithm:nonzeroCIS}
	\textbf{Input:} ASHG $(N,\vf)$ with $n$ agents, 	size parameters $\lb, \ub$, partition size $k$\\
	\textbf{Output:} $(\lb,\ub)$-partition consisting of $k$ coalitions
	\begin{algorithmic}[1]
		\If{$n > k \ub$ or $k > \left\lfloor \frac n{\lb}\right\rfloor$}
		\Return ``There exists no $\boundpair$-partition of $N$ into $k$ coalitions.''\label{ln:nofeas}
		\EndIf
		\State $A\leftarrow N$ \Comment{Still available agents}
		\State $x \leftarrow n-\lb k$\Comment{Number of agents in coalitions of size more than $\lb$}
		\For {$i\in [k]$}\Comment{\textbf{Phase I}: Formation of good coalitions for arbitrary leaders}
		\State Select any agent $\genA\in A$
		\State $C_i\leftarrow \{\genA\}\cup \topset{\lb-1}{A}$\Comment{Best coalition of minimum size}
		\State $C_i\leftarrow \{\genA\}\cup \topset{\min\{\ub-\lb,x\}}{\text{Fr}(\genA,A\setminus C_i)}$\Comment{Add friends}
		\State $A\leftarrow A\setminus C_i$
		\If{$|C_i| > \lb$}
		\State $x \leftarrow x - (|C_i| - \lb)$
		\EndIf
		\EndFor
		\While {$x>0$}\Comment{\textbf{Phase II}: Fill coalitions in inverse order}
		\State Let $i^* := \arg\max\{i\in [k]\colon |C_i| < \ub\}$.
		\State Choose $B\subseteq A$ with $|B| = \min \{x,\ub-|C_{i^*}|\}$
		\State $C_{i^*} \leftarrow C_{i^*} \cup B$
		\State $x \leftarrow x - |C_{i^*}|$
		\EndWhile\\
		
		\Return $\costr = \{C_1,\ldots, C_k\}$
	\end{algorithmic}
\end{algorithm}

\begin{proof}
	Assume that we are given an ASHG with nonzero valuations.
	Given a size bounds $\lb$ and $\ub$, and a target partition size $k$, we run \Cref{algorithm:nonzeroCIS}.
	This algorithm starts by testing whether there exists a $(\lb,\ub)$-partition consisting of $k$ coalitions by testing the conditions of \Cref{prop:feasibleKpart}.
	Hence, it (correctly) terminates in line~\ref{ln:nofeas} if and only if no such partition exists.
	
	Assume that such a partition exists.
	By definition, each coalition in this partition has to contain at least $\lb$ agents, while there exists $n-\lb k$ further agents that can be distributed to the $k$ coalitions.
	Throughout the algorithm, the variable $x$ keeps track of how many further agents we can add to form coalitions of size larger than $\lb$.
	\Cref{algorithm:nonzeroCIS} operates in two stages: in the first stage, we form coalitions of size at least $\lb$ by choosing a \emph{leader} and adding their $\lb - 1$ most-preferred agents.
	We further increase this coalition by adding friends to the leader, while being careful to not form coalitions of too large size or allocating more than a further $x$ agents.
	Intuitively, we end this stage with a partition of optimal coalitions for the leaders subject to available agents and forming coalitions of size at least $\lb$.
	In the second phase, we add the remaining agents to existing coalitions in the inverse order of their creation.
	This ensures that earlier leaders are treated preferentially, so they can still not deviate while later leaders could only deviate to nonfull coalitions, which will not accept them.
	
	We will now formally show the correctness of the algorithm.
	As argued above, the algorithm returns a partition if and only if any $\boundpair$-partition into $k$ coalitions exists.
	From now on assume that this is the case.
	Since such a partition exists, it is possible to form $k$ coalitions with $\lb$ agents in Phase~I, as long as at most $n-\lb k$ agents are additionally added to larger coalitions.
	Clearly, the algorithm has this property.
	Moreover, since there is a partition into $k$ coalitions, the coalitions formed in Phase~I have enough space for the remaining agents to join in Phase~II.
	We conclude that the algorithm produces a $\boundpair$-partition into $k$ coalitions.
	
	It remains to prove that the returned partition $\costr = \{C_1,\dots, C_k\}$ is a CIS\feastab partition.
	Let $i^*$ be the index of the last coalition to which agents are added in Phase~II, where we set $i^* = k + 1$ if Phase~II was not entered (because we had $x = 0$ in the beginning of Phase~II).
	
	Consider any agent $\genA\in N$ and let $i\in [k]$ with $\genA\in C_i$.
	Consider an index $j\in [k]$ with $j\neq i$.
	We will show that there exists no $\boundpair$-feasible CIS deviation by $\genA$ to $C_j$.
	For this, we can assume that $C_j$ is not full as otherwise no such deviation exists.
	Assume first that $j < i$ and let $\genB$ be the leader of $C_j$.
	If $\vf_{\genB}(\genA) < 0$, then $\genA$ cannot perform a CIS deviation to join $C_j$.
	Otherwise, since valuations are nonzero, we have that $\vf_{\genB}(\genA) > 0$, i.e., $\genA$ is a friend of $\genB$.
	However, $\genA$ was not added as an additional friend even though their was space in the coalition.
	Hence, it must have been the case that $x\le \mu-\lambda$ when adding further agents.
	Thus, after the formation of $C_j$, $x$ was updated to be $0$.
	Therefore, no larger-indexed coalition can be of size larger than $\lb$ and we must have that $|C_i| = \lb$. 
	Hence, there exists no $\boundpair$-feasible deviation by $\genA$.
	
	Now assume that $j > i$.
	Since $C_j$ is not full and has a higher index than $i$, we know that no agents were added to $C_i$ in Phase~II.
	Assume first that $\genA$ is not the leader of $C_i$.
	If $|C_i| = \lb$, then there exists no $\boundpair$-feasible CIS deviation by any agent in $C_i$.
	Otherwise, the leader of $C_i$ added friends after adding the $\lb - 1$ best agents.
	Hence, $\genA$ is a friend of the leader in $C_i$ and denied to perform a CIS deviation.
	Finally, assume that $\genA$ is the leader of $C_i$.
	Since $i < j$, all agents in $C_j$ where available, i.e., contained in $A$, when $C_i$ was formed.
	Hence, $C_i$ must yield a better utility than the best coalition $C'$ of size at least $\lb$ for $\genA$ with agents in $C_j$.
	We have that $\uf_{\genA}(C_i) \ge \uf_{\genA}(C'\cup\{\genA\})\ge \uf_{\genA}(C_j\cup\{\genA\})$.
	Hence, $\genA$ cannot form an NS, and therefore CIS, deviation to join $C_j$.
	Altogether, we conclude that \Cref{algorithm:nonzeroCIS} returns a CIS\feastab $\boundpair$-partition.
	\end{proof}

Next, we consider ASHGs $(N,\vf)$ with nonnegative valuations, i.e., $\vf_i(j)\ge 0$ for all $i,j\in N$.
Hence, an agent would never deny an agent to enter their coalition and we, therefore have that CIS\feastab is equal to CNS\feastab, i.e., there is only a constraint on abandoning a coalition.
The algorithmic idea is similar to \Cref{algorithm:CISub}, however, we have to ensure that there are enough agents to reach a size of $\lb$ for all $k$ coalitions, which further restricts how many friends an agent might take to join a coalition.

\begin{restatable}{theorem}{nonnegCIS}\label{thm:nonnegCIS}
	For ASHGs with nonnegative valuations, there exists a polynomial-time algorithm that computes a CIS\feastab $(\lb,\ub)$-partition into $k$ coalitions, or determines that no such partition exists.
\end{restatable}

\begin{algorithm}[tb]
	\caption{CIS\feastab $\boundpair$-partitions in ASHGs with nonnegative valuations.}\label{algorithm:nonnegCIS}
	\textbf{Input:} ASHG $(N,\vf)$ with $n$ agents, 	size parameters $\lb, \ub$, partition size $k$\\
	\textbf{Output:} $(\lb,\ub)$-partition consisting of $k$ coalitions
	\begin{algorithmic}[1]
		\If{$n > k \ub$ or $k > \left\lfloor \frac n{\lb}\right\rfloor$}
		\Return ``There exists no $\boundpair$-partition of $N$ into $k$ coalitions.''\label{ln:nofeas2}
		\EndIf
		\State $A\leftarrow N$ \Comment{Still available agents}
		\State $C_i \leftarrow \emptyset$ for $i\in [k]$\Comment{Initialize coalitions}
		\State $x \leftarrow n-\lb k$\Comment{Number of agents in coalitions of size more than $\lb$}
		\While{$A\neq \emptyset$}
		\State Select any agent $\genA\in A$\label{ln:leaders}
		\State $h\leftarrow 0$\Comment{Highest utility}
		\State $z\leftarrow 1$\Comment{Index of best coalition}
		\State $B \leftarrow \emptyset$\Comment{Best added agents}
		\For{$i\in [k]$}
		\State $r\leftarrow \min\{x + \max\{0,\lb - |C_i|\},\ub-|C_i|\}$\Comment{Number of agents allowed to add to $C_i$}
		\State $h' \leftarrow \uf_{\genA}(\{\genA\}\cup C_i\cup \topset{r-1}{\text{Fr}(a,A)})$\Comment{Best utility when joining $C_i$ with friends}
		\If{$h'>h$}
		\State $h \leftarrow h'$
		\State $z\leftarrow i$
		\State $B \leftarrow \topset{r}{\text{Fr}(a,A)}$
		\EndIf
		\EndFor
		\State $x \leftarrow x - \max\{0,1 + |B| - \max \{0,\lb - |C_z|\}\}$
		\State $C_z \leftarrow C_z \cup \{\genA\}\cup B$\label{ln:coalupdate}
		\State $A\leftarrow A\setminus (\{\genA\}\cup B)$
		\EndWhile\\
		
		\Return $\costr = \{C_1,\ldots, C_k\}$
	\end{algorithmic}
\end{algorithm}

\begin{proof}
	Assume that we are given an ASHG with nonzero valuations.
	Given a size bounds $\lb$ and $\ub$, and a target partition size $k$, we run \Cref{algorithm:nonnegCIS}.
	As in \Cref{algorithm:nonzeroCIS}, it starts by testing whether there exists a $(\lb,\ub)$-partition consisting of $k$ coalitions by testing the conditions of \Cref{prop:feasibleKpart}, and (correctly) terminates in line~\ref{ln:nofeas2} if and only if no such partition exists.
	
	If we know that a partition of the desired structure exists, the algorithm finds it as follows.
	Among the remaining agents, a new \emph{leader} is chosen arbitrarily.
	Then, for each of the $k$ coalitions, we have to decide, joining which of them leads to the highest utility.
	There, we have to be careful about not violating the upper bound on the coalition size and by not enlarging a coalition that is larger than $\lb$ by more than the number $x$ of agents that we may additionally add to coalitions.
	
	As in \Cref{algorithm:nonzeroCIS}, the parameter $x$ is initially set to the total number of agents that can added to the $k$ coalitions beyond their minimum sizes of $\lb$.
	Throughout the algorithm, when we reason to add agents to a specific coalition $C_i$ in the for loop within the while loop, we first determine the number $r$ equal to the maximum number of agents that we can add to $C_i$ such that we add at most $x$ agents apart from the $\lb - |C_i|$ agents if $|C_i| < \lb$.
	This number is equal to $x + \max\{0,\lb - |C_i|$\}, where the maximum ensures that we may still add $x$ agents if $|C_i|\ge \lb$.
	Then, before we actually add agents to a best coalition $C_z$ in \Cref{ln:coalupdate}, we decrease $x$ by $\max\{0,1 + |B| - \max \{0,\lb - |C_z|\}\}$.
	This quantity is equal to how much adding $\{\genA\}\cup B$, i.e., $1+|B|$ agents, might exceed the space in $C_z$ beyond reaching a size of $\lb$.
	If $1 + |B| < \max \{0,\lb - |C_z|\}$, i.e., if adding $\{\genA\}\cup B$ would lead to a coalition of size at most $\lb$, we do not want to decrease $x$.
	This is ensured by the outer maximum.
	Hence, throughout the execution of the algorithm, the number $x$ captures the number of agents that may still be added to coalitions of size at least $\lb$ to ensure that all other coalitions are guaranteed at least $\lb$ agents.
	This implies that \Cref{algorithm:nonnegCIS} returns a $\boundpair$-partition into $k$ coalitions.
	
	It remains to prove that the returned partition $\costr = \{C_1,\dots, C_k\}$ is a CIS\feastab partition.
	We refer to all agents selected in \Cref{ln:leaders} of the algorithm as \emph{leaders}.
	Now, every agent in the final partition that is not a leader is the friend of a leader and hence denied to perform any CIS deviation.
	
	Now consider a leader $\genA$ that is assigned to coalition $C_i$ for $i\in [k]$.
	Consider $j\in [k]$ with $i\neq j$, i.e., the index of a different coalition.
	If $|C_j| = \ub$ or $|C_i| = \lb$, then it is not possible to join $C_j$ by a $\boundpair$-feasible deviation, so we may assume that $|C_j| < \ub$ and $|C_i| > \lb$.
	Let $\hat C_j\subseteq C_j$ be the subset of $C_j$ of agents that were in $C_j$ in the iteration of the while loop when $\genA$ was the active leader.
	Moreover, let $\hat x$ and $\hat A$ be the value of $x$ and $A$ in the beginning of this iteration.
	Let $\hat r$ be the value of $r$ when $C_j$ was considered in the for loop.
	Finally, let $\hat C_i\subseteq C_i$ be the subset of $C_i$ of agents that were in $C_i$ after the coalition update in \Cref{ln:coalupdate} in this iteration.
	
	Since $x$ only decreases throughout the execution of the algorithm, and decreases by one for each agent that we add to $\hat C_j$ letting this coalition exceed a size of $\lb$, we know that $|C_j|\le |\hat C_j| + \hat r$ and $C_j\setminus \hat C_j \subseteq \hat A$.
	As $|C_j| < \ub$, we also have that $|C_j\setminus \hat C_j| \le \ub - |\hat C_j| - 1$.
	Moreover, since we decreased $x$ by the amount of agents in $C_j$ exceeding $\lb$ compared to $\hat C_j$ and we have additionally decreased $x$ after this iteration due to adding an agent to $C_i$ (either when $\genA$ and their friends joined or, since $|C_i|> \lb$, at a later stage).
	Hence, we know that $|C_j\setminus \hat C_j| \le \hat x - 1 + \max\{0,\lb - |\hat C_j|\}$, where the $-1$ accounts for the decrease of $x$ due to $C_i$.
	Together, we know that $|C_j\setminus \hat C_j| \le \hat r - 1$.
	Hence, it holds that 
    
	\begin{equation*}
		\uf_{\genA}(\{\genA\}\cup \hat C_j\cup \topset{\hat r-1}{\text{Fr}(a,\hat A)} \ge \uf_{\genA}(\{\genA\}\cup C_j)\text.
	\end{equation*}
	
	Moreover, since we were forming $\hat C_i$ when $\genA$ was leader, we know that this must have yielded at least the same utility, i.e., $\uf_{\genA}(\hat C_i) \ge \uf_{\genA}(\{\genA\}\cup \hat C_j\cup \topset{\hat r-1}{\text{Fr}(a,\hat A)}$, since $\hat C_i$ was formed.
	Finally, since all valuations are nonnegative, we have that $\uf_{\genA}(C_i) \ge \uf_{\genA}(\hat C_i)$.
	Combining all three inequalities yields $\uf_{\genA}(C_i) \ge\uf_{\genA}(\{\genA\}\cup C_j)$.
	Hence, $\genA$ cannot improve their utility by joining $C_j$.
	We conclude that $\costr$ is a CIS\feastab $\boundpair$-partition.
\end{proof}

\section{Conclusion}

In this paper, we have investigated stability based on deviations by single agents in additively separable hedonic games, where the output coalitions are restricted to be within given size bounds.
Since size bounds are ubiquitous in coalition formation applications, we think that our paper presents a valuable step to understand the tractability of coalition formation.

Interestingly, the existence of stability can only be guaranteed when the deviating agent is only allowed to abandon a coalition that is larger than the required minimum size.
Otherwise, we cannot even guarantee our weakest stability concept of contractual individual stability.
This is an interesting contrast to the unconstrained case, where CIS deviations are Pareto improvements.
While they are also Pareto improvements in our setting, they may, however, leave the space of feasible partitions.
Hence, CIS partitions may not exist when all Pareto-optimal outcomes of the unconstrained game do not satisfy the size bounds.
Once we restrict solution concepts to feasible deviations, we have the existence of CIS\feastab partitions, and have presented a polynomial-time algorithms to extract them when the lower bound is equal to~$1$, or when valuations are nonzero or nonnegative.

We complement our analysis of existence properties by an investigation of computational complexity.
There, we present a complete picture when the lower bound is equal to~$1$.
In addition to the polynomial-time solvability of CIS, the only tractable case is CNS when the upper bound on coalitions size is~$2$.
For a nontrivial lower bound, we present an almost complete picture for NS, leaving only the case of coalitions bounded between~$2$ and~$3$ open.
Further open problems concern the complexity of IS and CNS for a lower bound of at least~$2$, or CIS\feastab on the complete domain of valuations.
Other interesting avenues for future work within our bounded-coalition framework include investigating other classes of hedonic games, such as fractional hedonic games \cite{ABB+17a}, or other solution concepts such as popularity \cite{ABS11c,BrBu20a}.

\section*{Acknowledgements}
	Most of this work was done when Martin Bullinger and Edith Elkind were at the University of Oxford.
	Martin Bullinger and Edith Elkind were supported by the AI Programme of The Alan Turing Institute. 
	Adam Dunajski was supported by the University of Edinburgh School of Mathematics.
	Matan Gilboa was supported by the Engineering and Physical Sciences Research Council under grant EP/W524311/1.

\appendix

\section*{Appendix}

In the appendix, we present further material, such as proofs missing from the main paper.

\section{Algorithm of Aziz et al.~\cite{ABS11c}}\label{app:CISflawed}

In this section, we consider the algorithm by \citet[Algorithm~1]{ABS11c} aimed at computing CIS partitions for ASHGs without size constraints.
We present an example, where this algorithm unfortunately fails.

\begin{example}\label{ex:CISflawed}
	Consider an ASHG $(N,\vf)$ where $N = \{\genA_1,\genA_2,\genA_3,\genA_4\}$ and valuations are given as
	\begin{itemize}
		\item $\vf_{\genA_1}(\genA_2) = \vf_{\genA_1}(\genA_4) = -1$,
		\item $\vf_{\genA_3}(\genA_1) = 3$, $\vf_{\genA_3}(\genA_2) = \vf_{\genA_3}(\genA_4) = 2$,
		\item $\vf_{\genA_4}(\genA_2) = 1$, and 
		\item all other valuations are~$0$.
	\end{itemize}
	
	An illustration is provided in \Cref{fig:CISflawed}.
	
	Assume that the first two leaders are $\genA_1$ and $\genA_2$.
	As they both do not receive positive value from any agent, this yields the first two coalitions $S_1 = \{\genA_1\}$ and $S_2 = \{\genA_2\}$.
	The next leader is $\genA_3$.
	Their first choice is to form a new coalition with their still available friend $\genA_4$, obtaining a utility of $2$.
	However, they may also join $S_1$ receiving a utility of $3$ or $S_2$ receiving a utility of $2$.
	Note that $\genA_1$ and $\genA_2$ would both approve of that because they both have a valuation of~$0$ for $\genA_3$.
	Since the best option among these is joining $S_1$, $\genA_3$ will join this coalition, updating $S_1 = \{\genA_1,\genA_3\}$.
	Now, as $\vf_{\genA_1}(\genA_4) = -1$, no latecomers join this coalition.
	
	Finally, the last leader is $\genA_4$.
	Their best option is joining $S_2$, which is approved of by $\genA_2$.
	Hence, the outcome of the algorithm by     \citet{ABS11c} is $\costr = \{\{\genA_1,\genA_3\},\{\genA_2,\genA_4\}\}$.
	However, this is not a CIS partition, because $\genA_3$ has a CIS deviation to join $\{\genA_2,\genA_4\}$.\hfill$\lhd$
\end{example}

\begin{figure}[tb]
	\centering
	\begin{tikzpicture}[every node/.style={draw, circle, minimum size=.7cm, inner sep=0pt}]
		\pgfmathsetmacro\figscale{2.5}
		\node (a1) at (0,0) {$\genA_1$};
		\node (a2) at (\figscale,0) {$\genA_2$};
		\node (a3) at (0,\figscale) {$\genA_3$};
		\node (a4) at (\figscale,\figscale) {$\genA_4$};
		
		\draw[->] (a1) edge node[midway,fill = white, draw = none] {$-1$} (a2);
		\draw[->] (a1) edge node[pos =.3,fill = white, draw = none] {$-1$} (a4);
		\draw[->] (a3) edge node[midway,fill = white, draw = none] {$3$} (a1);
		\draw[->] (a3) edge node[pos =.3,fill = white, draw = none] {$2$} (a2);
		\draw[->] (a3) edge node[midway,fill = white, draw = none] {$2$} (a4);
		\draw[->] (a4) edge node[midway,fill = white, draw = none] {$1$} (a2);
	\end{tikzpicture}
	\caption{Illustration of the ASHG of \Cref{ex:CISflawed}. 
		We only depict the nonzero valuations.
		\label{fig:CISflawed}}
\end{figure}
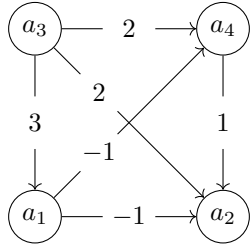

\section{Hardness results for $\lb = 1$}\label{app:upperbound}

In this section, we present the proofs of our hardness results when the lower bound satisfies $\lb = 1$.
We start with the consideration of contractual Nash stability. 

\CNSugeThree*

\begin{proof}
	We provide a reduction from the \NP-complete problem \etc{} (\etcs) \cite{GaJo79a}.
	An instance of \etcs{} is a pair $\langle R,\etcS{}\rangle$, where $R$ is a ground set of size $3\rho$ and $\etcS{}$ is a collection of $3$-element subsets of $R$; it is a Yes-instance if and only if there exists a subset $\etcS{}'\subseteq \etcS{}$ that partitions~$R$, i.e., covers $R$ and satisfies $|{\etcS{}}'|=\rho$.
	
	Recall that we illustrated our reduction in \Cref{fig:CNS-general} in the body of the paper.
	It works for any upper bound on coalition sizes $\ub$ with $\ub \ge 3$.
	Let $\langle R,\etcS{}\rangle$ be an instance of \etcs{} and set $z = 3|\etcS{}|-|R|$ (this is the sum of cardinalities of sets in $\etcS{}$ if removing some exact cover). 
	Without loss of generality, we may assume that $z\ge 0$. 
	
	Before formally describing our construction, we provide a high-level explanation.
		Elements of the ground set $R$ are represented by a gadget corresponding to an ASHG with $4$ agents not containing a CNS partition \cite{SuDi07b}.
	The sets $S$ in $\etcS{}$ are represented as follows: 
	For each element $r\in S$, there exist two agents denoted by $a_r^S$ and $\bar a_r^S$. 
	Moreover, there exist four more agents forming another auxiliary gadget as the one representing elements of $R$.
	Agents of the type $a_r^S$ have the ability to eliminate the instability caused by the gadget representing $r$. 	Agents of type $\bar a_r^S$ face a tension between two options. 
	The only agent that they receive a positive utility from is $a_r^S$, so they would like to be in a coalition with this agent.
	However, $\bar a_r^S$ is the only agent able to eliminate the instability caused by the corresponding auxiliary gadget of $S$.
	Hence, to guarantee stability, they are not allowed to form a coalition with $a_r^S$ or have an incentive to join them.
	The only way to enforce this is to group all agents of the type $a_r^S$ that are not interacting with $R$ gadgets in triplets according to the elements in $S$.
	Thus, either all or none of the agents of a set in $\etcS{}$ interact with an $R$ gadget, which yields a disjoint cover of $R$.
	
	We now formally define the reduced ASHG $(N,\vf)$ as follows. 
	Let $N = N_R \cup N_{\etcS{}}$ where
	\begin{itemize}
		\item $N_R = \cup_{r\in R}N_r$ with $N_r = \{\alpha_r,\beta_r,\gamma_r,\zeta_r\}$ for $r\in R$ and
		\item $N_{\etcS{}} = \cup_{S\in \etcS{}}N_S$ with $N_S = \{a^S_r,\bar a^S_r,\alpha^S_r,\beta^S_r,\gamma^S_r,\zeta^S_r\colon r \in S\}$ for $S\in \etcS{}$.
	\end{itemize}
	For $S\in \etcS{}$ and $r\in R$, we write $N^S_r = \{\alpha^S_r,\beta^S_r,\gamma^S_r,\zeta^S_r\}$.
	These agents will form an auxiliary gadget corresponding to a specific element in a specific set of $\etcS{}$.
	
	We now define the valuations $\vf$.
	First let $r\in R$.
	\begin{itemize}
		\item For $\genGrkA \in \{\alpha,\beta,\gamma\}$, we set $\vf_{\genGrkA_r}(\zeta_r) = 1$.
		\item For each pair $(\genGrkA,\genGrkB)\in \{(\alpha,\beta),(\beta,\gamma),(\gamma,\alpha)\}$, we set $\vf_{\genGrkA_r}(\genGrkB_r) = 0$.
	\end{itemize}
	Now let $S\in\etcS{}$ and $r\in S$.
	\begin{itemize}
		\item For $\genGrkA \in \{\alpha,\beta,\gamma\}$, we set $\vf_{\genGrkA^S_r}(\zeta^S_r) = 1$.
		\item For each pair $(\genGrkA,\genGrkB)\in \{(\alpha,\beta),(\beta,\gamma),(\gamma,\alpha)\}$, we set $\vf_{\genGrkA^S_r}(\genGrkB^S_r) = 0$.
		\item We set $\vf_{a^S_r}(\zeta_r) = \vf_{\zeta_r}(a^S_r) = 0$.
		\item We set $\vf_{\bar a^S_r}(\zeta^S_r) = \vf_{\zeta^S_r}(\bar a^S_r) = 0$.
		\item For $r'\in S\setminus\{r\}$, we set $\vf_{a^S_r}(a^S_{r'}) = 0$.
		\item We set $\vf_{\bar a^S_r}(a^S_r) = 2$ and for $r'\in S\setminus\{r\}$, we set  $\vf_{\bar a^S_r}(a^S_{r'}) = -1$.
		\item All valuations that have not been defined are set to $-3$.
	\end{itemize}
	
	The final negative value is chosen to be small enough to negate any positive valuation that an agent may have otherwise.
	Note that each agent has at most one positive valuation for another agent, and the maximum value of a positive valuation is~$2$.
	
	We claim that $\langle R,\etcS{}\rangle$ is a Yes-instance if and only if $(N,\vf)$ admits a CNS partition.
	
	\paragraph{$\implies$} Assume first that $\langle R,\etcS{}\rangle$ is a Yes-instance, i.e., there exists a subset $\etcS{}'\subseteq \etcS{}$ that partitions $R$.
			We define a partition $\costr$ as the union of the following coalitions:
	
	\begin{itemize}
		\item For every $r\in R$, we form $\{\alpha_r\}$, $\{\beta_r\}$, and $\{\gamma_r\}$.
		\item For $S\in \etcS{}', r \in S$, we form $\{a^S_r,\zeta_r\}$.
		\item For $S\in \etcS{}\setminus \etcS{}'$, we form $\{a^S_r\colon r \in S\}$.
		\item For $S\in \etcS{}$ and $r \in S$, we form $\{\alpha^S_r\}$, $\{\beta^S_r\}$, $\{\gamma^S_r\}$, and $\{\bar a^S_r,\zeta^S_r\}$.
	\end{itemize}
	
	The nonsingleton coalitions of $\costr$ are highlighted in blue in \Cref{fig:CNS-general}.
	
	Clearly, this is a $(1,\ub)$-partition for any $\ub\ge 3$.
	We argue that the constructed partition is contractually Nash-stable by performing a case analysis going through all agent types.
	
	\begin{itemize}
		\item For $r\in R$ and $\genGrkA \in \{\alpha,\beta,\gamma\}$, it holds that $\uf_{\genGrkA_r}(\costr) = 0$. 
		Moreover, joining any other coalition cannot improve their utility. 
		In particular, $\uf_{\genGrkA_r}(\costr(\zeta_r)\cup\{\genGrkA_r\}) = -2$.
		\item For $r\in R$, we have $\uf_{\zeta_r}(\costr) = 0$, which is the highest utility that $\zeta_r$ can achieve in any coalition.
		\item Similarly, for $S\in\etcS{}$ and $r\in S$, we have $\uf_{a^S_r}(\costr) = 0$ (regardless of whether $S\in\etcS{}'$ or not), which is the highest utility that $a^S_r$ can achieve in any coalition.
		\item For $S\in\etcS{}$, $r\in S$, and $\genGrkA \in \{\alpha,\beta,\gamma\}$, it holds that $\uf_{\genGrkA^S_r}(\costr) = 0$. 
		Moreover, joining any other coalition cannot improve their utility. 
		In particular, $\uf_{\genGrkA^S_r}(\costr(\zeta^S_r)\cup\{\genGrkA^S_r\}) = -2$.
		\item Let $S\in \etcS{}$ and $r\in S$.
		Then, $\uf_{\bar a^S_r}(\costr) = 0$.
		Moreover, joining any other coalition cannot improve their utility.
		Indeed, the only agent that $\bar a^S_r$ has a positive utility for is $a^S_r$.
		However, if $S\in \etcS{}'$, then $\{a^S_r,\zeta_r\}\in \costr$ and $\uf_{\bar a^S_r}(\costr(a^S_r)\cup\{\bar a^S_r\}) = -2$.
		In addition, if $S\in \etcS{}\setminus \etcS{}'$, then $\{a^S_r\colon r \in S\}\in \costr$ and $\uf_{\bar a^S_r}(\costr(a^S_r)\cup\{\bar a^S_r\}) = 0$.
		Note that the latter deviation is only $(\lb,\ub)$-permissive if $\ub \ge 4$, but it is excluded in any case.
	\end{itemize}
	
	We conclude that $\costr$ is contractually Nash-stable.
	
	\paragraph{$\impliedby$}
	Conversely, assume that $(N,\vf)$ contains a $(\lb,\ub)$-partition~$\costr$ that is contractually Nash-stable.
	We start by understanding the structure of $\costr$.
	The first step is to determine the coalition of $\zeta_r$ for $r\in R$.
	\begin{claim}\label{cl:zetaCostr}
		Let $r\in R$.
		Then there exists $S\in \etcS{}$ with $r\in S$ such that $\{\zeta_r,a_r^S\}\in \costr$.
	\end{claim}
	\begin{claimproof}
		Let $r\in R$.
		Consider $\genGrkA \in \{\alpha,\beta,\gamma\}$.
		Then it holds that $\costr(\genGrkA_r)\subseteq N_r$. 
		This follows because there is no agent with a positive valuation for $\genGrkA_r$ and that can, therefore, prevent $\genGrkA_r$ from abandoning their coalition. 
		Therefore, as $\genGrkA_r$ receives a negative utility in any coalition containing an agent outside $N_r$, they would then abandon their coalition to form a singleton coalition.
		Hence $\costr(\genGrkA_r)\subseteq N_r$.
		
		Define $C := \costr(\zeta_r)$.
		Our next goal is to determine the composition of $C$.
		Assume for contradiction that $C\subseteq N_r$.
		If $|C| \ge 3$, then $C$ contains some agent among $\{\alpha_r,\beta_r,\gamma_r\}$ that receives a negative utility.
		For example, if $\{\alpha_r,\beta_r\}\subseteq C$, then $\uf_{\beta_r}(\costr) \le -2$.
				Since no agent has a positive valuation for $\alpha_r$, $\beta_r$, and $\gamma_r$, such an agent could then perform a CNS deviation to form a singleton coalition, a contradiction.
		Hence, it must be the case that $|C| \le 2$.
		By symmetry, we may assume without loss of generality that $C\subseteq \{\zeta_r,\alpha_r\}$.
		But then $\uf_{\gamma_r}(\costr) \le 0$ while $\uf_{\gamma_r}(C\cup\{\gamma_r\}) \ge 1$.
		Thus, $\gamma_r$ could perform a CNS deviation to join $\zeta_r$. 
		Note that this yields a coalition of size at most $3\le \ub$, so it is permissible. 
		Hence, we also obtain a contradiction in this case, and conclude that $C\not\subseteq N_r$.
		
		As we already know that $\costr(\genGrkA_r)\subseteq N_r$ for $\genGrkA \in \{\alpha,\beta,\gamma\}$, we also conclude that $C\cap N_r = \{\zeta_r\}$.
		Hence, there is no agent in $C$ that gains a positive valuation from $\zeta_r$.
		Consequently, $\uf_{\zeta_r}(\costr) \ge 0$ as they could otherwise perform a CNS deviation to form a singleton coalition.
		Thus, since $C\not\subseteq N_r$, there exists $S\in \etcS{}$ with $r\in S$ such that $a_r^S\in C$ and $C$ can only contain such agents apart from $\zeta_r$.
		To show that there is only one such agent in $C$, let $T\in \etcS{}$ with $r\in T$ and $T\neq S$.
		If $a^T_r\in C$, then $\uf_{a^T_r}(\costr) \le -3$. 
		Then, since $\bar a^T_r\notin C$, $a^T_r$ could perform a CNS deviation to form a singleton coalition, a contradiction.
		We conclude that 		$C = \{\zeta_r,a_r^S\}$, showing the claim.
	\end{claimproof}
	
	Similarly, we can uniquely determine the coalition of $\zeta_r^S$ for $S\in \etcS{}$ and $r\in S$.
	
	\begin{claim}\label{cl:zetaScostr}
		Let $S\in \etcS{}$ and $r\in S$.
		Then it holds that $\{\zeta_r^S,\bar a_r^S\}\in \costr$.
	\end{claim}
	\begin{claimproof}
		The proof of this claim is analogous to the proof of \Cref{cl:zetaCostr}.
		In fact, once we know that $\costr(\zeta_r^S)\not\subseteq N^S_r$, $\costr(\zeta_r^S)\cap N^S_r = \{\zeta^S_r\}$, and $\uf_{\zeta^S_r}(\costr) \ge 0$, we can immediately conclude that $\{\zeta_r^S,\bar a_r^S\}\in \costr$, because $\bar a_r^S$ is the only agent for which $\zeta^S_r$ has a nonnegative valuation.
	\end{claimproof}
	
	We are ready to elicit a certificate for $\langle R,\etcS{}\rangle$ being a Yes-instance of \etcs{}.
	Define $\etcS{}' = \{S\in \etcS{}\colon \costr(a^S_r)\cap N_R \neq \emptyset \text{ for some } r \in S\}$. 
	By \Cref{cl:zetaCostr}, we know that $\etcS{}'$ covers $R$.
	It remains to show that no element is covered twice. Therefore, let $r\in R$ and consider $S\in \etcS{}'$ with $\{\zeta_r,a_r^S\}\in \costr$.
	Let $T\in \etcS{}$ with $r\in T$ and $T\neq S$. We have to show that $T\notin \etcS{}'$.
	
	Define $C:= \costr(a_r^T)$.
	By \Cref{cl:zetaScostr}, we know that $\bar a_r^T\notin C$. 	Hence, it must hold that $\uf_{a_r^T}(\costr)\ge 0$ as there is no agent that could prevent $a_r^T$ from deviating to form a singleton coalition.
	Moreover, since $\{\zeta_r,a_r^S\}\in \costr$, we know that $\zeta_r\notin C$ and, therefore, $C\subseteq \{a^T_{r'}\colon r' \in T\}$.
	Note that, by \Cref{cl:zetaScostr}, $\uf_{\bar a_r^T}(\costr) = 0$ and there is no agent that has a positive valuation for $\bar a_r^T$.
	Hence, if $|C|\le 2$, then $\uf_{\bar a_r^T}(C\cup \{\bar a_r^T\}) \ge 1$, and $\bar a_r^T$ can perform a CNS deviation joining $C$, which yields a coalition of size at most $3\le \ub$.
	It follows that $|C|\ge 3$ and, therefore, $C = \{a^T_{r'}\colon r' \in T\}$.
	Hence, there exists no $r'\in R$ such that $\costr(a^T_{r'})\cap N_R\neq \emptyset$.
	We conclude that $T\notin \etcS{}'$, as desired.
	\end{proof}

We continue with Nash and individual stability.

\NSandIS*

\begin{proof}
	We provide a reduction from the \NP-complete problem \mmm{} (\mmms) \cite{HoKi93a}.
	An instance of \mmms{} is a pair $\langle G, k\rangle$ where $G$ is a (undirected and unweighted) graph and $k$ is a positive integer; it is a Yes-instance if and only if there exists a maximal\footnote{Given a graph $G = (V,E)$, a matching $M\subseteq E$ is maximal if for each edge in $E$, some of its endpoints is covered by $M$. In other words, $M$ is inclusion-maximal among matchings.} matching in $G$ of size at most $k$.
	It is known that \mmms{} is \NP-complete, even for bipartite graphs $G = (A\cup B,E)$ where $|A| = |B|$ \cite{HoKi93a,Aziz13a}.
	
	Let $\langle G, k \rangle$ be an instance of \mmms{} where $G = (A\cup B,E)$ and $|A| = |B| = n$.
	We define a reduced ASHG $(N,\vf)$, where $N = A\cup B\cup X$, where $A$ and $B$ represent the vertices of the graph in the source instance, and $X = \{x_i^j\colon i\in [n-k], 0\le j\in [5]\}$ are auxiliary agents.
	We define valuations as
	\begin{itemize}
		\item $\vf_{a}(b) = \vf_b(a) = 3$ for $a\in A$ and $b\in B$ with $\{a,b\}\in E$,
		\item $\vf_{a}(x_i^1) = \vf_{x_i^1}(a) = 2$ for all $i\in [n-k]$,
		\item $\vf_{x_i^{j}}(x_i^{j'}) = 2$ and $\vf_{x_i^{j'}}(x_i^{j}) = 1$ for $i\in [n-k]$ and $j,j'\in [5]$ with $j'-j \equiv_5 1$, and
		\item all other valuations are $-6n$.
	\end{itemize}
	
	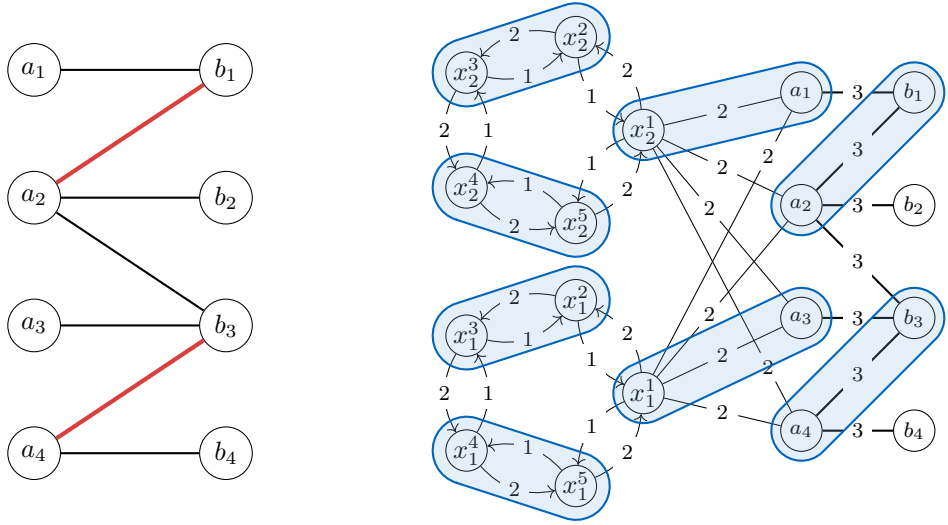
\begin{figure*}[tb]
		\centering
		\begin{tikzpicture}[element/.style={draw, circle, minimum size=.7cm, inner sep=0pt},smallnode/.style={draw, circle, minimum size=.55cm, inner sep=0pt}]
			\pgfmathsetmacro\xsscale{1.5}
			\pgfmathsetmacro\xscale{1.7}
			\pgfmathsetmacro\circrad{1.3}
			\foreach[count = \k] \j in {3,2,1,0}{
				\node[element] (a\k) at (0,\j*\xscale){$a_{\k}$};
				\node[element] (b\k) at (1.5*\xscale,\j*\xscale){$b_{\k}$};
			}
			\foreach \i/\j in {1/1,2/2,2/3,3/3,4/4}{
				\draw[thick] (a\i) edge (b\j);
			}
			\foreach \i/\j in {2/1,4/3}{
				\draw[ultra thick, myred] (a\i) edge (b\j);
			}
			
			\foreach[count = \k] \j in {3,2,1,0}{
				\node[smallnode] (a\k) at ($(6*\xscale, .3) + (0,\j*\xsscale)$) {\footnotesize $a_{\k}$};
				\node[smallnode] (b\k) at ($(6*\xscale, .3) + (\xsscale,\j*\xsscale)$) {\footnotesize $b_{\k}$};
			}
			\foreach \i/\j in {1/1,2/1,2/2,2/3,3/3,4/3,4/4}{
				\draw[thick] (a\i) edge node[midway, fill = white] {\footnotesize $3$} (b\j);
			}
			\foreach[count = \k] \i in {.8,4.3}{
				\node (p\k) at (4*\xscale, \i){};
				\foreach \l/\j in {1/0,2/72,3/144,4/216,5/288}{
					\node[smallnode] (c\k\l) at ($(p\k) + (\j:\circrad)$) {$x_{\k}^{\l}$};
				}
				\foreach \l/\j in {1/2,2/3,3/4,4/5,5/1}{
					\draw[bend right, ->] (c\k\l) edge node[midway, fill = white] {\footnotesize $2$} (c\k\j);
					\draw[bend right, ->] (c\k\j) edge node[midway, fill = white] {\footnotesize $1$} (c\k\l);
				}
				\draw[thick,myblue, fill=myblue!50, fill opacity=0.2] \convexpath{c\k2, c\k3}{0.45cm};
				\draw[thick,myblue, fill=myblue!50, fill opacity=0.2] \convexpath{c\k4, c\k5}{0.45cm};
			}
			
			\foreach \a/\k in {1/2,2/2,3/1,4/1}{
				\draw (a\a) edge node[midway, circle, fill = white, inner sep=2pt] {\footnotesize $2$} (c\k1);
			}
			\draw (a2) edge node[pos = .6, circle, fill = white, inner sep=2pt] {\footnotesize $2$} (c11);
			\draw (a3) edge node[pos = .6, circle, fill = white, inner sep=2pt] {\footnotesize $2$} (c21);
			\draw (a4) edge node[pos = .17, circle, fill = white, inner sep=2pt] {\footnotesize $2$} (c21);
			\draw (a1) edge node[pos = .17, circle, fill = white, inner sep=2pt] {\footnotesize $2$} (c11);

			\foreach \i/\j in {2/1,4/3}{
				\draw[thick,myblue, fill=myblue!50, fill opacity=0.2] \convexpath{a\i, b\j}{0.4cm};
			}
			\foreach \i/\j in {1/2,3/1}{
				\draw[thick,myblue, fill=myblue!50, fill opacity=0.2] \convexpath{a\i, c\j1}{0.4cm};
			}
		\end{tikzpicture}
		\caption{Illustration of the proof of \Cref{thm:NSandIS1u}. 
			On the left, we depict a bipartite graph with vertex set $A\cup B$ where $A = \{a_i\colon i\in [4]\}$ and $B = \{b_i\colon i\in [4]\}$.
			We ask whether there exists a maximal matching of size at most~$2$.
			On the right, we depict the reduced ASHG.
			The edge weights indicate the valuations, where an undirected edge means a mutual valuation. Omitted edges have weight $-6n$.
			The matching $\{\{a_2,b_1\}, \{a_4,b_3\}\}$ indicated with thick red edges is  maximal and of size~$2$.
			It corresponds to the Nash (and individually) stable partition indicated by its blue nonsingleton coalitions.
		}
		\label{fig:NSandISreduction}
	\end{figure*}
		An illustration of the game and the correspondence of Yes-instances established in the below proof is provided in \Cref{fig:NSandISreduction}.
	The idea is to integrate the \mmms{} instance into an ASHG by making agents corresponding to vertices from the same side of the bipartition incompatible via introducing a large negative weight.
	Moreover, we add $n-k$ gadgets on agent sets $\{x_i^j\colon j\in [5]\}$, where $i\in [n-k]$, corresponding to a well-known game not containing an individual partition for ASHGs with unbounded coalition sizes \cite[Example~$5$]{BoJa02a}.
	It is easy to see that this game does not contain an IS $(1,\ub)$-partition for any $\ub \ge 2$.
	Hence, some agent in these gadgets has to form a coalition with an agent outside of the gadget.
	For this to happen for all $n-k$ gadgets, we need $n-k$ agents from $A$, only allowing $k$ of them to mimic a matching.
	
	We proceed with a formal proof of correctness of our reduction.    
	Let $\ub \ge 2$.
	We will simultaneously deal with Nash and individual stability by proving the equivalence of the following statements:
	\begin{enumerate}
		\item $G$ contains a maximal matching of size at most $k$.\label{it:MMM}
		\item $(N,\vf)$ contains an NS $(1,\ub)$-partition.\label{it:NS}
		\item $(N,\vf)$ contains an IS $(1,\ub)$-partition.\label{it:IS}
	\end{enumerate}

	First, since Nash stability implies individual stability, we have that Statement~\ref{it:NS} implies Statement~\ref{it:IS}.
	We prove the remaining equivalence in the next two claims.
	
	\begin{claim}
		If $G$ contains a maximal matching of size at most $k$, then $(N,\vf)$ contains an NS $(1,\ub)$-partition.
	\end{claim}
	\begin{claimproof}
		Assume that $M$ is a maximal matching for $G$ of size at most $k$.
		Hence, there exists a subset $\{a_1,\dots, a_{n-k}\}\subseteq A$ of $n-k$ vertices of $A$ not covered by $M$.
		
		Consider the partition $\costr$ consisting of the following coalitions:
		\begin{itemize}
			\item If $e \in M$, we form the coalition $e$ (i.e., the coalition of the corresponding agents representing the vertices in $e$).
			\item For $i\in [n-k]$, we form $\{a_i,x_i^1\}$, $\{x_i^2,x_i^3\}$, and $\{x_i^4,x_i^5\}$.
			\item All other agents (these are agents in $B$ and, if $|M| < k$, also agents in $A$) are assigned to singleton coalitions.
		\end{itemize}
		
		Clearly, $\costr$ is a $(1,2)$-partition and, therefore, a $(1,\ub)$-partition.
		We claim that $\costr$ is Nash stable.
		Note that $\costr$ is individually rational and, therefore, Nash deviations would have to yield strictly positive utility.
		
		Clearly, no agent can increase their utility by joining a coalition containing an agent for which they have a negative valuation, because their utility after joining would be negative.
		Hence, no agent can benefit from joining a coalition of size~$2$ as these contain an agent for which they receive a negative utility.
		
		It remains to exclude deviations towards singleton coalitions.
		These can only be deviations by an agent $a\in A$ (or $b\in B$) joining an agent $b\in B$ (or $a\in A$) that is in a singleton coalition.
		This can only yield a positive utility if $\{a,b\}\in E$.
		Moreover, if the deviation is performed by an agent in $A$, they have to be in a singleton coalition or in a coalition with an agent in $X$ as the deviation would otherwise not increase their utility.
		For the same reason, if the deviation is performed by an agent in $B$, this agent has to be in a singleton coalition in $\costr$.
		In both cases, both $a$ and $b$ are not covered by $M$.
		Hence, a deviation towards a singleton deviation implies the existence of an edge $\{a,b\}\in E$ of agents not covered by $M$.
		However, by the maximality of $M$, such an edge cannot exist.
		Hence, there is no such deviation and $\costr$ is Nash-stable.
	\end{claimproof}
	
	\begin{claim}
		If $(N,\vf)$ contains an IS $(1,\ub)$-partition, then $G$ contains a maximal matching of size at most~$k$.
	\end{claim}
	\begin{claimproof}
		Assume that $(N,\vf)$ contains an IS $(1,\ub)$-partition $\costr$.
		Then, since deviations to form a singleton coalition are $(1,\ub)$-permissible, $\costr$ has to be individually rational.
		This implies that no agent is in a coalition with an agent for which they receive a negative utility.
		Indeed, the sum of all positive valuation towards other agents of agents in $A$, $B$, and $X$ is at most $3n + 2(n-k) < 5n$, $3n$, and $2n +3$, respectively, while a single negative valuation is $-6n$.
		
		Hence, as any two friends of an agent are enemies, this implies that all coalitions in $\costr$ are of size at most $2$.
		Consider $M = \{C\in \costr\colon A\cap C \neq \emptyset, B\cap C\neq \emptyset\}$.
		Clearly, $M$ is a matching in $M$, containing coalitions of size~$2$ from $\costr$ that consist of an agent in $A$ and an agent in $B$.
		We will show that $M$ is a maximal matching of size at most $k$.
		
		Assume for contradiction that there exists an edge $e = \{a,b\}\in E$ with $a\in A$, $b\in B$, and $a$ and $b$ are both not covered by $M$.
		Then, $a$ is in a singleton coalition or a coalition with an agent in $X$, while $b$ is in a singleton coalition.
		Hence, $a$ can perform a $(1,\ub)$-permissible deviation to join $\{b\}$, which increases their utility from at most $2$ to $3$, while it increases the utility of $b$.
		Hence, this would be an IS deviation, contradicting individual stability of $\costr$.
		
		Next, fix $i\in [n-k]$.
		By individual rationality, for $2\le j\le 5$, 		$x_i^j$ can only be in a coalition with agents in $\{x_i^{j'}\colon j'\in [5]\}$.
		Assume for contradiction that the same is true for $x_i^1$.
		Then, since all coalitions in $\costr$ are of size at most $2$, there exists an index $j\in [5]$ such that $x_i^j$ is in a singleton coalition.
		Consider $x_i^{j'}$ for $j'\in [5]$ with $j-j'\equiv_5 1$.
		We have $\uf_{x_i^{j'}}(\{x_i^{j},x_i^{j'}\}) = 2$, while $\uf_{x_i^{j'}}(\costr) \le 1$.
		Indeed, in $\costr$, $x_i^{j'}$ has to be in a singleton coalition or in a coalition with $x_i^{j''}$ where $j''-j'\equiv_5 1$, which would yield a utility of~$1$.
		Moreover, $\uf_{x_i^{j}}(\{x_i^{j},x_i^{j'}\}) = 1 > 0 = \uf_{x_i^{j}}(\costr)$.
		Hence, the deviation by $x_i^{j'}$ joining $x_i^{j}$ would be an $(1,\ub)$-permissible IS deviation, which cannot exist.
		Hence, $x_i^1$ must be in a coalition of size~$2$ with an agent outside $\{x_i^{j}\colon j\in [5]\}$.
		The only such agents are agents in $A$.
		As this is true for every $i\in [n-k]$, there must exist a set of $n-k$ agents in $A$ that are not in a coalition with agents in $B$ and, hence, their corresponding vertices in $G$ are uncovered by $M$.
		We conclude that $|M| \le k$ as each edge in $M$ must cover a vertex in $A$.
		\end{claimproof}
	
	We have proved the equivalence of all three statements. Clearly, this yields correctness of the reduction for both Nash and individual stability.
	\end{proof}

\section{Proof of \Cref{thm:NSuge4}}\label{app:bothbounds}

In this section, we present the proof of our for hardness result for Nash stability under a nontrivial lower bound on the coalition sizes.

\NSugefour*

\begin{proof}
	We provide a reduction from \etcs{}, as defined in the beginning of the proof of \Cref{thm:CNS-Uge3}. 
	Fix size bounds $\lb$ and $\ub$ with $\ub\ge 4$ and $\lb < \ub$. 
	Given an instance $\langle R,\etcS{}\rangle$ of \etcs{}, we construct an ASHG $(N,\vf)$ as follows. 
	Set $\etcpara := |\etcS{}|-\frac{|R|}{3}$. 
	This is the number of sets remaining in $\etcS$ after removing an exact cover.
	
	Let $N=B\cup E\cup T \cup D \cup \{\alpha\}$ where
	\begin{itemize}
		\item $B=\{\beta_r\}_{r\in R}$,
		\item $E=\bigcup_{S\in\etcS{}}E_S$ where $E_S=\{\xi_S^i\}_{i=1}^{\ub-3}$ for each $S\in\etcS{}$,
		\item $T=\bigcup_{i=1}^{\etcpara}T_i$, where $T_i=\{t_i^1, t_i^2, t_i^3\}$ for each $i\in [\etcpara]$ ($T$ for `triplets'), and
		\item $D$ satisfies $|D|=\lceil\frac{\lb-1}{\ub-\lb}\rceil\ub+\ub$ ($D$ for `dummy').
	\end{itemize}
	
	We refer to agents from set $X\in \{B,E,T,D\}$ as $X$-agents.
	Moreover, we call the $B$-, $E$-, and $T$-agents \emph{core agents}, and the rest \emph{structure agents}. 
	A specific $D$-agent will typically be denoted by $d$.
	If $r\in S$ for some $S\in\etcS{}$, we say that $\beta_r$ \emph{corresponds} to $S$ ($\beta_r$ may correspond to multiple sets $S$). 
	Similarly, we say that all agents in $E_S$ \emph{correspond} to $S$, and for any $i\in [\etcpara]$ we say that the three agents in the set $T_i$ \emph{correspond} to each other.
	Notice that the number of structure agents is $|D|+1=\lceil\frac{\lb-1}{\ub-\lb}\rceil\ub+\ub+1$, which is sufficient for them to be structured into feasible coalitions among themselves, according to \Cref{prop:LargeFeasible}.
	Furthermore, one may verify that $|N|=(\lceil\frac{\lb-1}{\ub-\lb}\rceil+|\etcS{}|+1)\cdot \ub+1$; In particular, since $\ub\geq 4$ and thus $\ub\neq 1$, we have that $\ub$ does not divide $|N|$.
	
	We proceed to describing the valuations $v$ of the agents. 
	\begin{itemize}
		\item For any core agent $c\in B\cup E\cup T$, we set $\vf_c(\alpha)=-\ub$.
		\item Let $r\in R$. We set $\vf_{\beta_r}(t)=\vf_{\beta_r}(\xi^i_S)=-1$ for any $t\in T$, $i\in [\ub-3]$, and $S \in \etcS{}$ such that $r\notin S$.
		\item Let $S\in\etcS{}$ and $i\in [\ub-3]$. We set $\vf_{\xi^i_S}(\xi^j_{S'})=\vf_{\xi^i_S}(\beta_r)=-1$ for any $j\in [\ub-3]$, $S' \in \etcS{}$ such that $S'\neq S$, and $r\in R$ such that $r\notin S$.
		\item Let $i\in[\etcpara]$ and $j\in\{1,2,3\}$. We set $\vf_{t^j_i}(t^{j'}_{i'})=\vf_{t^j_i}(\beta_r)=-1$ for any $i'\in[\etcpara]$ with $i'\neq i$, $j\in\{1,2,3\}$, and $r\in R$.
		\item Let $d\in D$. We set 
		\begin{itemize}
			\item $\vf_d(c)=-1$ for any $c\in B\cup E\cup T$, and
			\item $\vf_d(\alpha)=\ub$.
		\end{itemize}
		\item We set $\vf_{\alpha}(c)=1$ for any $c\in B\cup E\cup T$.
		\item Any valuation not specified so far is set to $0$.
		We note that all agents thus assign value $0$ to any dummy agent.   
	\end{itemize}
	An illustration of the main components of the reduction can be found in \Cref{fig:NSuge4_S_gadget}.

	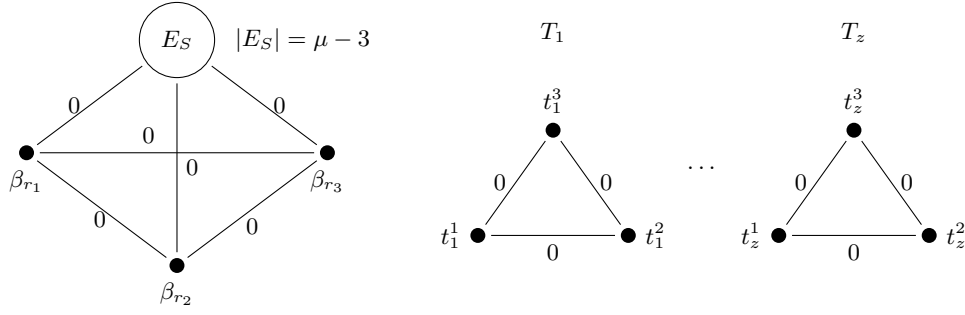
\begin{figure*}[tb]
		\centering    
		\begin{tikzpicture}[every node/.style={font=\small}]
			\begin{scope}
				\node[draw, circle, minimum size=1cm] (ES) at (0,0) {$E_S$};
				\node[anchor=west, xshift=4pt] at (ES.east) {$|E_S| = \ub - 3$};
				
				\node[circle, fill, inner sep=2pt, label=below:$\beta_{r_1}$] (B1) at (-2,-1.5) {};
				\node[circle, fill, inner sep=2pt, label=below:$\beta_{r_2}$] (B2) at (0,-3) {};
				\node[circle, fill, inner sep=2pt, label=below:$\beta_{r_3}$] (B3) at (2,-1.5) {};
				
				\draw[-, shorten >=2pt, shorten <=2pt] (ES) -- (B1) node[pos=0.5, left] {0};
				\draw[-, shorten >=2pt, shorten <=2pt] (ES) -- (B2) node[pos=0.5, right] {0};
				\draw[-, shorten >=2pt, shorten <=2pt] (ES) -- (B3) node[pos=0.5, right] {0};
				
				\draw[-, shorten >=2pt, shorten <=2pt] (B1) -- (B2) node[pos=0.6, left] {0};
				\draw[-, shorten >=2pt, shorten <=2pt] (B2) -- (B3) node[pos=0.5, below] {0};
				\draw[-, shorten >=2pt, shorten <=2pt] (B3) -- (B1) node[pos=0.6, above, sloped] {0};
			\end{scope}
			
			\begin{scope}[xshift=5cm]
				\node[circle, fill, inner sep=2pt, label=left:$t_1^1$] (T11) at (-1,-2.6) {};
				\node[circle, fill, inner sep=2pt, label=right:$t_1^2$] (T12) at (1,-2.6) {};
				\node[circle, fill, inner sep=2pt, label=above:$t_1^3$] (T13) at (0,-1.2) {};
				\node at (0,0.1) {$T_1$};
				
				\draw[-, shorten >=2pt, shorten <=2pt] (T11) -- (T12) node[pos=0.5, below] {0};
				\draw[-, shorten >=2pt, shorten <=2pt] (T12) -- (T13) node[pos=0.5, right] {0};
				\draw[-, shorten >=2pt, shorten <=2pt] (T13) -- (T11) node[pos=0.5, left] {0};
				
				\node at (2,-1.7) {$\dots$};
				
				\node[circle, fill, inner sep=2pt, label=left:$t_z^1$] (T21) at (3,-2.6) {};
				\node[circle, fill, inner sep=2pt, label=right:$t_z^2$] (T22) at (5,-2.6) {};
				\node[circle, fill, inner sep=2pt, label=above:$t_z^3$] (T23) at (4,-1.2) {};
				\node at (4,0.1) {$T_z$};
				
				\draw[-, shorten >=2pt, shorten <=2pt] (T21) -- (T22) node[pos=0.5, below] {0};
				\draw[-, shorten >=2pt, shorten <=2pt] (T22) -- (T23) node[pos=0.5, right] {0};
				\draw[-, shorten >=2pt, shorten <=2pt] (T23) -- (T21) node[pos=0.5, left] {0};
			\end{scope}
			
		\end{tikzpicture}
		\caption{Illustration of key aspects in the proof of \Cref{thm:NSuge4}. On the left is the gadget connecting the nodes of $E_S$ with their corresponding $B$-agents. On the right are the $T$-triangles. Undirected edges imply mutual valuations. Omitted edges crossing between $T$-triangles have weight $-1$.}
		\label{fig:NSuge4_S_gadget}
	\end{figure*}
	
	We provide a proof sketch before the formal proof. 
	First, note that the valuations' description creates a ``cat-and-mouse'' dynamics between core agents and agent $\alpha$. 
	Namely, we set the valuations such that agent $\alpha$ only likes core agents, while core agents do not like $\alpha$. 
	Using this we ensure that in any NS partition, on the one hand agent $\alpha$ is separated from the core agents (as otherwise the core agents would deviate away from $\alpha$), but on the other hand the coalitions containing core agents are full (as otherwise agent $\alpha$ would deviate to join them).
	Having this separation, we show that the only possibility for an NS partition is, for all $S\in \etcS$, letting $E_S$ form a coalition either with the three $B$-agents corresponding to $S$ or with some $T$-triplet. 
	This implies that $B$-agents must be partitioned by a subset $\etcS'\in\etcS$.
	
	An interesting feature of the reduction is that, since we may have a nontrivial lower bound on the coalition sizes, the lack of individual rationality is not sufficient to argue that a partition is not NS (an agent cannot deviate to a singleton coalition when $\lb>1$). 
	To overcome this, we introduce dummy agents, and ensure that there exists a nonfull coalition $C'$ consisting only of dummy agents. 
	Thus, individual rationality is restored, as any agent with negative utility would be able to deviate to $C'$. 
	We guarantee the existence of such a coalition by setting $|N|$ that is not divided by $\ub$, so there must exist a nonfull coalition; using the ``cat-and-mouse'' dynamics mentioned earlier, we can then show any nonfull coalition contains only dummies.

	We will now formally show that $\langle R,\etcS{}\rangle$ is a Yes-instance of $\etcs{}$ if and only if $(N,\vf)$ admits an NS partition.
	
	\paragraph{$\implies$} Assume first that $\langle R,\etcS{}\rangle$ is a Yes-instance, i.e., there exists a subset $\etcS{}'\subseteq \etcS{}$ that partitions $R$. We observe that $|\etcS{}'|=\frac{|R|}{3}$, by definition of a partition and since $|S|=3$ for each $S\in\etcS{}$.
	We define a partition $\costr$ as the union of the following coalitions:
	
	\begin{itemize}
		\item For each $S\in \etcS{}'$ we form $E_S\cup\{\beta_r\}_{r\in S}$.
		\item Let us arbitrarily enumerate all $S\in \etcS{}\setminus \etcS{}'$, denoting $S_1,...,S_{\etcpara}$. For each $i\in[\etcpara]$, we form $E_{S_i}\cup T_i$.
		\item We arbitrarily pick a set $D_{\alpha}\subseteq D$ with $|D_{\alpha}|=\ub-1$, and form $D_{\alpha}\cup\{\alpha\}$.
		\item The remaining $\frac{\lb-1}{\ub-\lb}\ub+1$ dummy agents are arbitrarily partitioned into feasible coalitions among themselves (by \Cref{prop:LargeFeasible}, this is well-defined).
	\end{itemize}
	
	Clearly, $\costr$ is a $\boundpair$-partition. 
	Furthermore, since $\etcS{}'$ partitions $R$, all $B$-agents are indeed assigned a (unique) coalition, and, therefore, $\costr$ is a well-defined partition of $N$. 
	Assume towards contradiction that there exists a $\boundpair$-permissible NS deviation by some agent $a$ joining a coalition $C\subseteq \costr$.
	Then it must be that $|C|< \ub$, as otherwise there is no room for $a$ to join it.
	Therefore, $C$ contains only  dummy agents, since any coalition in $\costr$ containing a nondummy agent is of size $\ub$.
	This implies that agent $a$ obtains a utility of $0$ in the partition resulting from the deviation.
	However, it may be verified that all agents obtain nonnegative utilities in $\costr$: $E$-agents only form coalitions with other $E$- or $B$-agents corresponding to the same set $S\in\etcS{}$, or with $T$-agents; $B$-agents only form coalitions with other $B$-agents, and with $E$-agents corresponding to the same set $S\in\etcS{}$; $T$-agents only form coalitions with their corresponding $T$-agents, and with $E$-agents; and $D$-agents only form coalitions with other $D$-agents or with agent $\alpha$.
	Therefore, $a$ does not benefit from the deviation, a contradiction.
	
	\paragraph{$\impliedby$}
	Conversely, assume that $(N,\vf)$ contains a $(\lb,\ub)$-partition~$\costr$ that is Nash-stable.
	First, because $\ub$ does not divide $|N|$, there must exist a coalition $C'\in \costr$ that is not full, i.e., $|C'| < \ub$.
	In the next claims, we want to understand the structure of $\costr$ with the aim of finding a subset $\etcS{}'\subseteq \etcS{}$ that partitions $R$.
	
	\begin{claim}\label{NSuge4:alpha_full}
		The coalition $\costr(\alpha)$ is full.
	\end{claim}
	\begin{claimproof}
		Assume for contradiction that $\costr(\alpha)$ is not full. 
		Since $|D|\geq \ub$, there must exist some $d\in D$ such that $\costr(d)\neq \costr(\alpha)$. 
		Hence, $\uf_d(\costr)\leq 0$. 		Consider the deviation of $d$ joining $\costr(\alpha)$. 
		Since $\costr(\alpha)$ is not full, this deviation is $\boundpair$-permissible.
		This will result in a utility of at least $1$ for $d$, since $\vf_d(\alpha)=\ub$ and $vf_d(a)\geq -1$ for any agent $a\in N\bs\{\alpha\}$ (and there can be at most $\ub - 2$ such agents in $\costr(\alpha)$). 
		Hence, we have a contradiction to the Nash stability of $\costr$.
	\end{claimproof}
	
	\begin{claim}\label{NSuge4:alpha_cores}
		The coalition $\costr(\alpha)$ contains no core agents.
	\end{claim}
	\begin{claimproof}
		Assume for contradiction that $\costr(\alpha)$ contains some core agent $c\in B\cup E\cup T$. 
		Then we have that $\uf_c(\costr)\leq -\ub$. 
		But then $d$ may deviate to the coalition $C'$ that is not full, and obtain utility at least $-\ub+1$, since $\alpha\notin C'$ (since $C'$ is not full while $\costr(\alpha)$ is, by \Cref{NSuge4:alpha_full}), and $c$ assigns value at least $-1$ to any other agent. 
		Thus, $\costr$ is not Nash-stable, a contradiction.
	\end{claimproof}
	
	\begin{claim}\label{NSuge4:core_full}
		Any coalition $C\in \costr$ that contains a core agent is full.
	\end{claim}
	\begin{claimproof}
		Assume for contradiction that $C$ is not full and contains some core agent $c\in B\cup E\cup T$. 
		By \Cref{NSuge4:alpha_cores} we have that $\uf_{\alpha}(\costr)=0$. 
		Hence, it is an NS deviation for $\alpha$ to deviate to $C$, a contradiction to the Nash stability of $\costr$.
	\end{claimproof}
	
	\begin{claim}\label{NSuge4:all_nonnegative}
		All agents obtain nonnegative utility in $\costr$.
	\end{claim}
	\begin{claimproof}
		By \Cref{NSuge4:core_full,NSuge4:alpha_full}, we have that $C'$ contains only dummy agents. Thus, all agents in $C'$ clearly obtain utility $0$. Furthermore, any agent in $N\bs C'$ obtaining a negative utility in $\costr$ would benefit from deviating to $C'$, a contradiction to the stability of $\costr$.     
	\end{claimproof}
		
	\begin{claim}\label{NSuge4:cores_0}
		Any core agent assigns value $0$ to all members of their coalition in $\costr$.
	\end{claim}
	\begin{claimproof}
		This follows from \Cref{NSuge4:all_nonnegative} and the fact that core agents do not assign a positive value to any agent.
	\end{claimproof}
	
	\begin{claim}\label{NSuge4:dummies_cores}
		Core agents and dummy agents cannot reside in the same coalition in $\costr$.
	\end{claim}
	\begin{claimproof}
		Assume for contradiction that there exists $d\in D$ such that $\costr(d)$ contains a core agent. 
		Then by \Cref{NSuge4:alpha_cores}, we have that $\alpha \notin\costr(d)$. 
		Furthermore, since dummy agents assign a negative value to core agents, and only assign a positive value to agent $\alpha$, we have that $d$ obtains a negative utility in $\costr$, a contradiction to \Cref{NSuge4:all_nonnegative}. 
	\end{claimproof}
	
	\begin{claim}\label{NSuge4:Es_together}
		Let $S\in\etcS{}$, and consider some agent $\xi_S^i\in E_S$. Then $E_S\subseteq\costr(\xi_S^i)$.
	\end{claim}
	\begin{claimproof}
		Recall that agent $\xi_S^i$ assigns value $0$ to all agents in $E_S$, all $T$-agents, the three $B$-agents corresponding to $S$, and no one else.
		Assume towards contradiction that $|E_S\cap\costr(\xi_S^i)|<|E_S|=\ub-3$.
		Then, by \Cref{NSuge4:core_full}, $\costr(\xi_S^i)$ must contain 
				$A\subseteq N\bs E_S$ with $|A|\ge 4$. 
		By \Cref{NSuge4:dummies_cores,NSuge4:alpha_cores}, agents in $A$ cannot be dummy agents or $\alpha$. 
		Thus, they are core agents, and specifically $B$- or $T$-agents (since noncorresponding $E$-agents assign negative to utility to each other, which would contradict \Cref{NSuge4:cores_0}). 
		Moreover, $A$ cannot include a combination of $B$- and $T$-agents, since those assign negative values to each other, which would contradict \Cref{NSuge4:cores_0}. 
		If all four are $B$-agents then $\xi_S^i$ assigns a negative value to at least one of them (since $\xi_S^i$ corresponds only to three $B$-agents), and if all four are $T$-agents then at least two of them assign a negative value to each other (since $T$-agents form triplets of corresponding agents). Thus, either way we have a contradiction to \Cref{NSuge4:cores_0}.
	\end{claimproof}
	
	\begin{claim}\label{NSuge4:Es_coalition}
		Let $S\in\etcS{}$ with $S=\{r_1,r_2,r_3\}\subseteq R$, and consider some agent $\xi_S^i\in E_S$. Then either $\costr(\xi_S^i)=E_S\cup \{\beta_{r_1},\beta_{r_2},\beta_{r_3}\}$ or $\costr(\xi_S^i)=E_S\cup T_j$ for some $j\in[\etcpara]$.
	\end{claim}
	
	\begin{claimproof}
		By \Cref{NSuge4:core_full,NSuge4:Es_together}, we have that $\costr(\xi_S^i)=E_S\cup A$ where $A=\{a_1,a_2,a_3\}$ for some agents $a_1,a_2,a_3\in N$. By \Cref{NSuge4:alpha_cores} we have that $\alpha\notin A$, by \Cref{NSuge4:dummies_cores} we have that $A\cap D=\emptyset$, and by \Cref{NSuge4:cores_0} we have that $A\cap E=\emptyset$ (since noncorresponding $E$-agents assign negative to utility to each other). Hence, we have that $A\subseteq B\cup T$. Further, since $B$- and $T$-agents assign negative utility to each other, by \Cref{NSuge4:cores_0} we have that $A\subseteq B$ or $A\subseteq T$. If $A\subseteq B$ but $A\neq\{\beta_{r_1},\beta_{r_2},\beta_{r_3}\}$ then $\xi_S^i$ must assign a negative value to some $a\in A$; and if $A\subseteq T$ but $a_1$, $a_2$, and $a_3$ do not all correspond to each other, then at least two of them assign a negative value to each other. 
		Either way, we have a contradiction to \Cref{NSuge4:cores_0}.
		Hence, we have that either $A=\{\beta_{r_1},\beta_{r_2},\beta_{r_3}\}$ or $A=T_j$ for some $j\in[\etcpara]$.
	\end{claimproof}
	
	We now have enough information to extract a subset $\etcS{}'\subseteq \etcS{}$ that partitions $R$. 
	By \Cref{NSuge4:Es_coalition}, each $S\in\etcS{}$ with $S=\{r_1,r_2,r_3\}$ induces a coalition containing $E_S$ and three additional agents $A=\{a_1,a_2,a_3\}$, such that $A=\{\beta_{r_1},\beta_{r_2},\beta_{r_3}\}$ or $A=T_j$ for some $j\in[\etcpara]$; call $S$ a \emph{solution set} if its induced coalition is of the former form (i.e. $A=\{\beta_{r_1},\beta_{r_2},\beta_{r_3}\}$), and an \emph{extra set} otherwise. 
	Crucially, for the latter option we have $j\in[\etcpara]$, namely there are at most $\etcpara$ extra sets, implying there are at least $\frac{|R|}{3}$ solution sets $S\in\etcS{}$. 
	Notice that solution sets are pairwise disjoint: Otherwise we have some $r\in S_1\cap S_2$ where $S_1,S_2\in\etcS{}$ are solution sets and $S_1\neq S_2$, but then by definition of solution sets we have that $\costr(r)$ contains agents from $E_{S_1}$ and $E_{S_2}$, a contradiction to \Cref{NSuge4:cores_0}. 
	Hence, there are exactly $\frac{|R|}{3}$ solution sets, since there are only $|R|$ agents in $B$. Thus, consider the subset $\etcS{}'\subseteq\etcS{}$ defined by $\etcS{}'=\{S'\in \etcS{}\colon S' \text{ is a solution set}\}$. 
	Since each set in $\etcS{}'$ contains three elements of $R$, and there are $\frac{|R|}{3}$ such sets, we have that $\etcS{}'$ covers $R$. Since additionally the sets in $\etcS{}'$ are pairwise disjoint, we have that $\etcS{}'$ forms a partition of $R$.
	\end{proof}


\begin{thebibliography}{32}
\providecommand{\natexlab}[1]{#1}
\providecommand{\url}[1]{\texttt{#1}}
\expandafter\ifx\csname urlstyle\endcsname\relax
  \providecommand{\doi}[1]{doi: #1}\else
  \providecommand{\doi}{doi: \begingroup \urlstyle{rm}\Url}\fi

\bibitem[Agarwal et~al.(2025)Agarwal, Agarwal, Raj, and Nath]{AARS25a}
Pulkit Agarwal, Harshvardhan Agarwal, Vaibhav Raj, and Swaprava Nath.
\newblock Harmonious balanced partitioning of a network of agents.
\newblock In \emph{Proceedings of the 24th International Conference on
  Autonomous Agents and Multiagent Systems (AAMAS)}, pages 41--49, 2025.

\bibitem[Aziz(2013)]{Aziz13a}
Haris Aziz.
\newblock Stable marriage and roommate problems with individual-based
  stability.
\newblock In \emph{Proceedings of the 12th International Conference on
  Autonomous Agents and Multiagent Systems (AAMAS)}, 2013.

\bibitem[Aziz and Savani(2016)]{AzSa15a}
Haris Aziz and Rahul Savani.
\newblock Hedonic games.
\newblock In Felix Brandt, Vincent Conitzer, Ulle Endriss, J.~Lang, and
  Ariel~D. Procaccia, editors, \emph{Handbook of Computational Social Choice},
  chapter~15. Cambridge University Press, 2016.

\bibitem[Aziz et~al.(2013)Aziz, Brandt, and Seedig]{ABS11c}
Haris Aziz, Felix Brandt, and Hans~Georg Seedig.
\newblock Computing desirable partitions in additively separable hedonic games.
\newblock \emph{Artificial Intelligence}, 195:\penalty0 316--334, 2013.

\bibitem[Aziz et~al.(2019)Aziz, Brandl, Brandt, Harrenstein, Olsen, and
  Peters]{ABB+17a}
Haris Aziz, Florian Brandl, Felix Brandt, Paul Harrenstein, Martin Olsen, and
  Dominik Peters.
\newblock Fractional hedonic games.
\newblock \emph{ACM Transactions on Economics and Computation}, 7\penalty0
  (2):\penalty0 1--29, 2019.

\bibitem[Ballester(2004)]{Ball04a}
Coralio Ballester.
\newblock {NP}-completeness in hedonic games.
\newblock \emph{Games and Economic Behavior}, 49\penalty0 (1):\penalty0 1--30,
  2004.

\bibitem[Bil{\`o} et~al.(2022)Bil{\`o}, Monaco, and Moscardelli]{BMM22a}
Vittorio Bil{\`o}, Gianpiero Monaco, and Luca Moscardelli.
\newblock Hedonic games with fixed-size coalitions.
\newblock In \emph{Proceedings of the 36th AAAI Conference on Artificial
  Intelligence (AAAI)}, pages 9287--9295, 2022.

\bibitem[Bogomolnaia and Jackson(2002)]{BoJa02a}
Anna Bogomolnaia and Matthew~O. Jackson.
\newblock The stability of hedonic coalition structures.
\newblock \emph{Games and Economic Behavior}, 38\penalty0 (2):\penalty0
  201--230, 2002.

\bibitem[Brandt and Bullinger(2022)]{BrBu20a}
Felix Brandt and Martin Bullinger.
\newblock Finding and recognizing popular coalition structures.
\newblock \emph{Journal of Artificial Intelligence Research}, 74:\penalty0
  569--626, 2022.

\bibitem[Brandt et~al.(2024)Brandt, Bullinger, and Tappe]{BBT23a}
Felix Brandt, Martin Bullinger, and Leo Tappe.
\newblock Stability based on single-agent deviations in additively separable
  hedonic games.
\newblock \emph{Artificial Intelligence}, 334:\penalty0 104160, 2024.

\bibitem[Bullinger and Gilboa(2025)]{BuGi25a}
Martin Bullinger and Matan Gilboa.
\newblock Settling the complexity of popularity in additively separable and
  fractional hedonic games.
\newblock In \emph{Proceedings of the 34th International Joint Conference on
  Artificial Intelligence (IJCAI)}, pages 3771--3779, 2025.

\bibitem[Bullinger and Kraiczy(2024)]{BuKr24a}
Martin Bullinger and Sonja Kraiczy.
\newblock Stability in random hedonic games.
\newblock In \emph{Proceedings of the 25th ACM Conference on Economics and
  Computation (ACM-EC)}, page 212, 2024.

\bibitem[Bullinger and Suksompong(2024)]{BuSu24a}
Martin Bullinger and Warut Suksompong.
\newblock Topological distance games.
\newblock \emph{Theoretical Computer Science}, 981:\penalty0 114238, 2024.

\bibitem[Bullinger et~al.(2024)Bullinger, Elkind, and Rothe]{BER24a}
Martin Bullinger, Edith Elkind, and J{\"o}rg Rothe.
\newblock Cooperative game theory.
\newblock In J{\"o}rg Rothe, editor, \emph{Economics and Computation: An
  Introduction to Algorithmic Game Theory, Computational Social Choice, and
  Fair Division}, chapter~3, pages 139--229. Springer, 2024.

\bibitem[Cohen and Agmon(2025)]{CoAg25a}
Saar Cohen and Noa Agmon.
\newblock Egalitarianism in online coalition formation (extended abstract).
\newblock In \emph{Proceedings of the 24th International Conference on
  Autonomous Agents and Multiagent Systems (AAMAS)}, pages 2475--2477, 2025.

\bibitem[Darmann et~al.(2022)Darmann, D{\"o}cker, Dorn, and
  Schneckenburger]{DDDS22a}
Andreas Darmann, Janosch D{\"o}cker, Britta Dorn, and Sebastian
  Schneckenburger.
\newblock Simplified group activity selection with group size constraints.
\newblock \emph{International Journal of Game Theory}, 51\penalty0
  (1):\penalty0 169--212, 2022.

\bibitem[Deligkas et~al.(2024)Deligkas, Eiben, Knop, and Schierreich]{DEKS24a}
Argyrios Deligkas, Eduard Eiben, Du{\v{s}}an Knop, and {\v{S}}imon Schierreich.
\newblock Individual rationality in topological distance games is surprisingly
  hard.
\newblock In \emph{Proceedings of the 33rd International Joint Conference on
  Artificial Intelligence (IJCAI)}, pages 2782--2790, 2024.

\bibitem[Deligkas et~al.(2025)Deligkas, Eiben, Ioannidis, Knop, and
  Schierreich]{DEI+25a}
Argyrios Deligkas, Eduard Eiben, Stavros~D. Ioannidis, Du{\v{s}}an Knop, and
  {\v{S}}imon Schierreich.
\newblock Balanced and fair partitioning of friends.
\newblock In \emph{Proceedings of the 39th AAAI Conference on Artificial
  Intelligence (AAAI)}, pages 13754--13762, 2025.

\bibitem[Dimitrov et~al.(2006)Dimitrov, Borm, Hendrickx, and Sung]{DBHS06a}
Dinko Dimitrov, Peter Borm, Ruud Hendrickx, and Shao~C. Sung.
\newblock Simple priorities and core stability in hedonic games.
\newblock \emph{Social Choice and Welfare}, 26\penalty0 (2):\penalty0 421--433,
  2006.

\bibitem[Dr{\`e}ze and Greenberg(1980)]{DrGr80a}
Jacques~H. Dr{\`e}ze and Joseph Greenberg.
\newblock Hedonic coalitions: Optimality and stability.
\newblock \emph{Econometrica}, 48\penalty0 (4):\penalty0 987--1003, 1980.

\bibitem[Elkind and Wooldridge(2009)]{ElWo09a}
Edith Elkind and Michael Wooldridge.
\newblock Hedonic coalition nets.
\newblock In \emph{Proceedings of the 8th International Conference on
  Autonomous Agents and Multiagent Systems (AAMAS)}, pages 417--424, 2009.

\bibitem[Fioravantes et~al.(2025)Fioravantes, Gahlawat, and Melissinos]{FGM25a}
Foivos Fioravantes, Harmender Gahlawat, and Nikolaos Melissinos.
\newblock Exact algorithms and lower bounds for forming coalitions of
  constrained maximum size.
\newblock In \emph{Proceedings of the 39th AAAI Conference on Artificial
  Intelligence (AAAI)}, pages 13847--13855, 2025.

\bibitem[Flammini et~al.(2021)Flammini, Monaco, Moscardelli, Shalom, and
  Zaks]{FMM+21a}
Michele Flammini, Gianpiero Monaco, Luca Moscardelli, Mordechai Shalom, and
  Shmuel Zaks.
\newblock On the online coalition structure generation problem.
\newblock \emph{Journal of Artificial Intelligence Research}, 72:\penalty0
  1215--1250, 2021.

\bibitem[Ganian et~al.(2023)Ganian, Hamm, Knop, Schierreich, and
  Such{\`y}]{GHK+23a}
Robert Ganian, Thekla Hamm, Du{\v{s}}an Knop, {\v{S}}imon Schierreich, and
  Ond{\v{r}}ej Such{\`y}.
\newblock Hedonic diversity games: A complexity picture with more than two
  colors.
\newblock \emph{Artificial Intelligence}, 325:\penalty0 104017, 2023.

\bibitem[Garey and Johnson(1979)]{GaJo79a}
Michael~R. Garey and David~S. Johnson.
\newblock \emph{Computers and Intractability: A Guide to the Theory of
  NP-Completeness}.
\newblock W. H. Freeman, 1979.

\bibitem[Horton and Kilakos(1993)]{HoKi93a}
Joseph~D. Horton and Kyriakos Kilakos.
\newblock Minimum edge dominating sets.
\newblock \emph{SIAM Journal of Discrete Mathematics}, 6\penalty0 (3):\penalty0
  375--387, 1993.

\bibitem[Levinger et~al.(2023)Levinger, Azaria, and Hazon]{LAH23a}
Chaya Levinger, Amos Azaria, and Noam Hazon.
\newblock Social aware coalition formation with bounded coalition size
  (extended abstract).
\newblock In \emph{Proceedings of the 22nd International Conference on
  Autonomous Agents and Multiagent Systems (AAMAS)}, pages 2667--2669, 2023.

\bibitem[Levinger et~al.(2024)Levinger, Hazon, Simola, and Azaria]{LHSA24a}
Chaya Levinger, Noam Hazon, Sofia Simola, and Amos Azaria.
\newblock Coalition formation with bounded coalition size.
\newblock In \emph{Proceedings of the 23rd International Conference on
  Autonomous Agents and Multiagent Systems (AAMAS)}, pages 1119--1127, 2024.

\bibitem[Li et~al.(2023)Li, Micha, Nikolov, and Shah]{LMNS23a}
Lily Li, Evi Micha, Aleksandar Nikolov, and Nisarg Shah.
\newblock Partitioning friends fairly.
\newblock In \emph{Proceedings of the 37th AAAI Conference on Artificial
  Intelligence (AAAI)}, pages 5747--5754, 2023.

\bibitem[Sung and Dimitrov(2007)]{SuDi07b}
Shao~C. Sung and Dinko Dimitrov.
\newblock On myopic stability concepts for hedonic games.
\newblock \emph{Theory and Decision}, 62\penalty0 (1):\penalty0 31--45, 2007.

\bibitem[Sung and Dimitrov(2010)]{SuDi10a}
Shao~C. Sung and Dinko Dimitrov.
\newblock Computational complexity in additive hedonic games.
\newblock \emph{European Journal of Operational Research}, 203\penalty0
  (3):\penalty0 635--639, 2010.

\bibitem[Woeginger(2013)]{Woeg13a}
Gerhard~J. Woeginger.
\newblock A hardness result for core stability in additive hedonic games.
\newblock \emph{Mathematical Social Sciences}, 65\penalty0 (2):\penalty0
  101--104, 2013.

\end{thebibliography}
\end{document}